\documentclass[11pt,a4paper]{article}
\usepackage[T1]{fontenc}
\usepackage[utf8]{inputenc}
\usepackage{algorithm}
\usepackage{algpseudocode}
\usepackage{multirow}
\usepackage{lmodern}
\usepackage{microtype}
\usepackage[a4paper,margin=2.5cm]{geometry}
\usepackage{amsmath,amssymb,amsthm,bm}
\usepackage{mathtools}
\usepackage{graphicx}
\usepackage{booktabs}
\usepackage{xcolor}
\usepackage{caption}
\usepackage[round,authoryear]{natbib}
\usepackage[hidelinks]{hyperref}
\usepackage{enumitem}
\usepackage{placeins}

\theoremstyle{thmstyleone}
\newtheorem{theorem}{Theorem}
\newtheorem{proposition}[theorem]{Proposition}
\newtheorem{corollary}[theorem]{Corollary}

\theoremstyle{thmstyletwo}
\newtheorem{remark}{Remark}

\theoremstyle{thmstylethree}
\newtheorem{definition}{Definition}
\newtheorem{assumption}{Assumption}

\date{}
\title{\bfseries A
unified framework for estimating direct causal effect under spatial
confounding and interference, with the R package \textsf{spaci}}

\author{Isqeel Ogunsola\textsuperscript{1,2} and Olatunji Johnson\textsuperscript{1}\\[0.3em]
\small \textsuperscript{1}Department of Mathematics, University of Manchester, Manchester, UK\\[0.2em]
\small \textsuperscript{2}Department of Statistics, Federal University of Agriculture, Abeokuta, Nigeria \\[0.2em]
\small Corresponding author: Isqeel Ogunsola
(\href{mailto:olatunji.johnson@manchester.ac.uk}{\texttt{isqeel.ogunsola@manchester.ac.uk}})}

\begin{document}

\maketitle

\begin{abstract}    
\noindent Estimating causal effects is challenging with observational data due
to the lack of exchangeability, and is further complicated with spatial data.
In real-life applications such as agricultural and environmental studies,
unobserved spatial factors (spatial confounding, SC) and interactions of units
in nearby locations (spatial interference, SI) commonly occur jointly, yet
existing parametric methods address them separately. We develop a unified framework that
treats both simultaneously. We first define the estimand precisely, the
average \emph{direct} effect of a unit's own treatment on the treated, holding
neighbourhood exposure at its observed level and give an identification
result. 
We then study two estimators. The matching estimator iDAPS
integrates neighbourhood exposure and spatial proximity into propensity score
matching with balance-optimised weight. We prove a \emph{variogram bias
bound} showing its confounding bias is controlled by the spatial variability
of the confounder at the matched distance; this justifies the composite
metric, yields a consistency result, and provides a computable diagnostic.
The doubly robust estimator recoverU$+$ augments the propensity and outcome
models with the recovered confounder and the exposure; we prove it is doubly
robust up to an explicit residual bias from the unrecoverable
confounder component. We show that independence-based standard errors are
anticonservative under spatial dependence; we provide spatial
heteroskedasticity and autocorrelation consistent (HAC), block-bootstrap,
and randomisation-based alternatives whose calibration we verify. Every theoretical result is validated numerically. Estimating the
effect of selective catalytic and non-catalytic reduction (SCR/SNCR) technologies on
ambient ozone reveals no evidence of ozone reduction, while a naive analysis
would have delivered a confidently wrong sign. All methods are implemented in the developed
open-source R package \textsf{spaci}.
\end{abstract}


\maketitle

\section{Introduction}\label{sec:intro}

Investigating the cause-and-effect relationships of interventions, treatments
or exposures on an outcome is of interest in many disciplines, such as
agriculture, medicine, epidemiology, economics, environment and public
health, and plays a significant role in decision-making and policy
formulation. In many recent studies, research questions are causal in nature
rather than mere description or association modelling
\citep{cacciarelli2025we, grace2025causal}. In causal
inference with spatial observational data, challenges such as spatial
confounding (violation of the unconfoundedness assumption) and spatial
interference arising from treatment spatial dependence are inevitable
\citep{papadogeorgou2019}. For example, in agricultural
studies, estimating the effect of fertiliser on crops where soil nutrients
vary spatially, without accounting for hidden unmeasured spatial variables
(spatial confounders) that also affect the crop, will bias the estimated
effect; and fertiliser or pest control applied in one plot can leach into and
affect the surrounding field (interference). There are numerous analogous
scenarios in epidemiology, climate science, education and social-programme
studies. 

Spatial confounding is sometimes misinterpreted as spatial
interference and vice versa, especially when both are present and only one is
accounted for, although the two phenomena are distinct
\citep{papadogeorgou2019, ogunsola2026disentangling}. SC is a missing variable problem, while spatial interference
reflects complex dependencies or interactions of treatment units
\citep{giffin2023generalized, papadogeorgou2023, aronow2017estimating}. In
the presence of these two phenomena, causal effects are biased if not
accounted for, and invariably result in misleading inferences
\citep{papadogeorgou2023, ogunsola2026disentangling}. Figure~\ref{fig:dag}
gives the directed acyclic graph (DAG) for the problem in two spatial
settings, including the treatments ($A_1, A_2$), spatial confounders
($U_1, U_2$), their spatial dependence ($U_1 - U_2$) and spatial interference
($A_1 \rightarrow Y_2$, $A_2 \rightarrow Y_1$); the right panel shows the
problem at the block level.

\begin{figure}[h]
\centering
\begin{minipage}{0.5\textwidth}\centering
  \includegraphics[width=\linewidth]{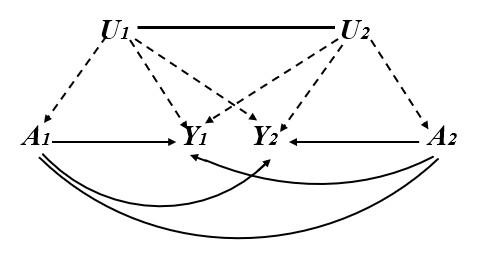}
\end{minipage}\hfill
\begin{minipage}{0.4\textwidth}\centering
  \includegraphics[width=0.8\linewidth]{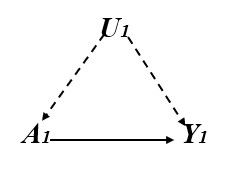}
\end{minipage}
\caption{The problem. \emph{Left:} causal effects in the presence of spatial
confounders $U_1, U_2$ (dashed: unobserved), spatial dependence and spatial
interference across two settings. \emph{Right:} the block-level view with
treatment $A$, outcome $Y$ and unmeasured confounder $U$.}\label{fig:dag}
\end{figure}

A substantial literature addresses each problem separately. For spatial
confounding, methods like restricted spatial regression  \citep{reich2006},
distance-adjusted propensity scores \citep[DAPS;][]{papadogeorgou2019},
structural equation approaches \citep{thaden2018}, machine-learning
adjustments \citep{gilbert2021}, and doubly robust estimation with partially
recovered confounders \citep[recoverU;][]{pokal2023improved} have been employed. Methods for addressing spatial interference include; exposure mappings and network estimators
\citep{aronow2017estimating, savje2021}, propensity-score methods
\citep{giffin2023generalized, zirkle2021addressing, forastiere2021}, and
instrumental variables \citep{giffin2021}. Methods that handle both
\emph{simultaneously} are scarce: \citet{papadogeorgou2023} propose a joint
Bayesian model whose conclusions are sensitive to prior specification, and
\citet{giffin2021} treated confounding and spillover via instruments that are
difficult to find and validate. A transparent, design-based frequentist
framework for the joint problem has been missing. This paper provides one,
with five contributions.

\emph{(i) A precise estimand and identification result}
(Section~\ref{sec:setup}). Under interference, ``the average treatment effect
on the treated'' is ambiguous, because switching a unit's treatment may also
change its neighbours' exposures. We define the \emph{average direct effect
on the treated at observed exposure} (ADET), show it is the quantity targeted
by both estimators, and prove identification under explicit assumptions,
including an honest \emph{recoverability} assumption isolating exactly what
must be believed about the unmeasured confounder.

\emph{(ii) A matching estimator with a bias guarantee}
(Section~\ref{sec:idaps}). iDAPS extends DAPS by matching on a composite of
propensity-score, spatial and neighbourhood-exposure distances, with weights
chosen by balance optimisation rather than manual tuning. Our key theoretical
result is a \emph{variogram bias bound}: a matched-pair difference is a
spatial differencing operation, and the residual confounding bias is
controlled by the variogram of the confounder at the matched distance. The
bound explains \emph{why} the composite metric works, yields a consistency
corollary, exposes an irreducible bias floor under micro-scale confounding,
and gives a diagnostic computable from data.

\emph{(iii) An honest guarantee for doubly robust estimation}
(Section~\ref{sec:recoveru}). recoverU$+$ extends recoverU method by including the
neighbourhood exposure in both nuisance models. Because only the
treatment-\emph{uncorrelated} part of the confounder can be recovered,
unqualified double robustness is impossible; we prove \emph{partial double
robustness}, with an explicit residual bias $b_{U_{A}}$ that vanishes when the
confounder does not drive treatment.

\emph{(iv) Valid, computation-aware inference} (Section~\ref{sec:inference}).
The influence values of the doubly robust estimator are spatially correlated,
so the usual independence-based standard error is anticonservative. We
provide spatial heteroskedasticity- and autocorrelation-consistent (HAC),
block-bootstrap and conditional-randomisation
alternatives and verify their finite-sample calibration.

\emph{(v) Validated software} (Section~\ref{sec:software}). All methods,
simulations and the application are reproducible with the R package
\textsf{spaci}; every proposition is accompanied by a numerical validation

Section~\ref{sec:sims} presents the simulation studies including an
explicit stress test of double robustness and of exposure-mapping
misspecification. In Section~\ref{sec:application}, the application to
SCR/SNCR technologies and ambient ozone formation is given. Section~\ref{sec:discussion} gives the discussion,
concludes with limitations, stated as results, and extensions.

\section{Setup, estimand and identification}\label{sec:setup}

\subsection{Data and interference structure}

Suppose in a spatial observational study we observe
$\{(Y_i, A_i, X_i, s_i)\}_{i=1}^{n}$ at locations
$s_i \in D \subset \mathbb{R}^2$: a continuous outcome $Y_i$, a binary
treatment $A_i \in \{0,1\}$, covariates $X_i \in \mathbb{R}^p$, and an
\emph{unmeasured} spatial confounder $U(s_i)$, a mean-zero random field.
Interference is summarised through a neighbourhood-exposure mapping
\begin{equation}\label{eq:exposure}
  E_i \;=\; \sum_{j \ne i} G_{ij}(\tau)\, A_j,
  \qquad
  G_{ij}(\tau) = \exp\!\big(-\lVert s_i - s_j \rVert / \tau\big),
  \quad G_{ii} = 0,
\end{equation}
optionally row-normalised so that $E_i$ is a weighted average of neighbours'
treatments; $\tau > 0$ is a spatial decay bandwidth. Potential outcomes are
indexed by a unit's own treatment and its exposure, $Y_i(a, e)$
\citep{hudgens2008, forastiere2021}.

We decompose the confounder field as
\begin{equation}\label{eq:decomp}
  U(s_i) \;=\; U_{A}(s_i) + U_{R}(s_i),
\end{equation}
where $U_{A} = \Pi(U \mid \mathcal{A})$ is the $L_2$-projection of $U$ onto
$\mathcal{A}$, the closed linear span of the (centred) treatment field
$\{A_j\}_{j=1}^{n}$ which is the component of the confounder correlated with
treatment and $U_{R}$ is the orthogonal remainder. This decomposition is central to
everything that follows, and the asymmetry between its two parts is a matter
of identification, not of technique. From a single realisation of the data,
any component of $U$ that moves with the treatment field is observationally
indistinguishable from a treatment effect: no procedure, however flexible,
can separate the two without further assumptions. The orthogonal component
$U_{R}$, by contrast, leaves a visible trace, smooth, spatially structured
variation in the outcome that treatment, covariates and exposure cannot
explain and can therefore be estimated from outcome residuals by
kriging \citep{pokal2023improved}. Recovered confounder method exploit
exactly this trace; what they cannot do, and what
Proposition~\ref{prop:pdr} makes precise, is recover $U_{A}$, whose influence
must be bounded rather than removed.

\subsection{Estimand}

\begin{definition}[Average direct effect on the treated at observed exposure]
\label{def:adet}
$\mathrm{ADET} = \mathbb{E}\big[ Y_i(1, E_i) - Y_i(0, E_i) \mid A_i = 1 \big]$.
\end{definition}

$\mathrm{ADET}$ is the effect of switching a treated unit's \emph{own} treatment off
while holding its realised neighbourhood exposure fixed: a \emph{direct}
effect. In the application of Section~\ref{sec:application} it answers a
concrete question: for a facility that installed the technology, what would
its ozone concentration have been had it not installed, with all
neighbouring installations held exactly as they are? Without interference,
$Y_i(a,e)$ reduces to $Y_i(a)$ and $\mathrm{ADET}$ coincides with the classical average treatment effect of the treated (ATT);
with interference the classical ATT is ill-defined, switching $A_i$ may
also change neighbours' exposures, so ``the effect of treatment'' silently
mixes the direct effect with spillovers effects and $\mathrm{ADET}$ is the well-defined
quantity both estimators target. The companion \emph{spillover effect on the
treated}, $\mathrm{SET}(e,e') = \mathbb{E}[Y_i(a,e) - Y_i(a,e') \mid A_i = a]$, is a
distinct estimand identified through the exposure margin of the outcome
model; keeping the two separate prevents estimators that in fact target
different quantities from being compared as if they estimated the same
thing (see the comparability remark in Section~\ref{sec:software}).

Throughout we use the working structural model
\begin{equation}\label{eq:outcome}
  Y_i \;=\; \theta_1 A_i + \theta_2^{\top} X_i + \theta_3 E_i
            + \theta_U\, U(s_i) + \varepsilon_i,
  \qquad \mathbb{E}[\varepsilon_i \mid X, E, U, A] = 0,
\end{equation}
under which $\mathrm{ADET} = \theta_1$, making the simulation ground truth
unambiguous.

\subsection{Assumptions and identification}

\begin{assumption}[Latent conditional ignorability]\label{a:ign}
$\{Y_i(a,e)\} \perp \!\!\! \perp A_i \mid X_i, E_i, U(s_i)$.
\end{assumption}
\begin{assumption}[Correctly specified exposure mapping]\label{a:exp}
$Y_i(a,e) \perp \!\!\! \perp \bm{A}_{-i} \mid E_i = e$.
\end{assumption}
\begin{assumption}[Positivity]\label{a:pos}
$0 < \mathbb{P}(A_i = 1 \mid X_i, E_i, U) < 1 $
\end{assumption}
\begin{assumption}[Consistency]\label{a:con}
$Y_i = Y_i(A_i, E_i)$.
\end{assumption}
\begin{assumption}[Recoverability]\label{a:rec}
$\mathbb{E}[Y_i(0,e) \mid X, E, U] = \mathbb{E}[Y_i(0,e) \mid X, E, U_{R}]$; equivalently,
under \eqref{eq:outcome}, $U_{A}$ enters $Y(0)$ only through the additive term
$\theta_U U_{A}$.
\end{assumption}

Assumptions~\ref{a:ign}--\ref{a:con} are the interference-aware analogues of
the standard causal conditions \citep{hudgens2008, forastiere2021};
Assumption~\ref{a:exp} states that all interference flows through the
exposure summary \eqref{eq:exposure} \citep[for the consequences of
misspecifying it, see][and Section~\ref{sec:misspec}]{savje2024}.
Assumption~\ref{a:rec} is the honest crux of any recovered-confounder method,
and Proposition~\ref{prop:pdr} quantifies the bias when it fails.

\begin{proposition}[Identification]\label{prop:id}
Let $V_i = (X_i, E_i, U_{R}(s_i))$.
(a) Under Assumptions~\ref{a:ign}--\ref{a:rec},
\[
  \mathrm{ADET} \;=\; \mathbb{E}\big\{\, \mathbb{E}[Y \mid A{=}1, V] - \mathbb{E}[Y \mid A{=}0, V]
              \,\big|\, A = 1 \big\}.
\]
(b) Without Assumption~\ref{a:rec}, the right-hand side differs from $\mathrm{ADET}$
by exactly the residual term $b_{U_{A}}$ of Proposition~\ref{prop:pdr}.
\end{proposition}
A complete proof, with each step explained, is given in
Appendix~\ref{app:id}.

In words, once the exposure summary and the recoverable part of the
confounder are conditioned on, treated and control units with the same $V$
are exchangeable, and the familiar regression, matching and weighting
machinery applies. Part~(b) is as useful as part~(a): it says the framework
degrades \emph{gracefully} when the recoverability assumption fails, the
identification error is not arbitrary but exactly the quantity that
Proposition~\ref{prop:pdr} characterises and that a sensitivity analysis
can probe.

\begin{remark}\label{rem:ps}
The adjustment probability $p(V) = \mathbb{P}(A_i = 1 \mid X_i, E_i, U_{R})$ is an
individual propensity score that also conditions on a neighbourhood treatment
summary, in the spirit of \citet{forastiere2021}; the additions here are the
continuous kernel-weighted exposure and the latent confounder handled by
recovery. Conditioning on $E_i$, a function of other units' treatments
is licensed by a joint (conditional-field) model of the treatment vector; see
\citet{tchetgen2021auto}. The two estimators operationalise this adjustment
probability differently: recoverU$+$ estimates $p(V)$ directly, augmenting
its propensity model with $E$ and the recovered confounder
(Section~\ref{sec:recoveru}), whereas iDAPS fits its working propensity
score on the measured covariates $X$ alone and controls exposure and
spatial proximity through its matching metric instead
(Section~\ref{sec:idaps}).
\end{remark}

\section{iDAPS: composite matching with a variogram bias bound}
\label{sec:idaps}

\subsection{The estimator}

Let $\hat p(X_i)$ be a working propensity score fitted on the measured
covariates, and define three pairwise distances, each rescaled to $[0,1]$:
$D^{PS}_{ij} = |\hat p(X_i) - \hat p(X_j)|$,
$D^{\mathrm{Sp}}_{ij} = \lVert s_i - s_j\rVert$, and
$D^{\mathrm{Int}}_{ij} = |E_i - E_j|$. iDAPS matches on the composite
\begin{equation}\label{eq:comp}
  D^{\mathrm{comp}}_{ij} \;=\; \pi_1 D^{PS}_{ij} + \pi_2 D^{\mathrm{Sp}}_{ij}
                  + \pi_3 D^{\mathrm{Int}}_{ij},
  \qquad \pi_k \ge 0,\ \ \textstyle\sum_k \pi_k = 1 .
\end{equation}
Setting $\pi_3 = 0$ recovers DAPS \citep{papadogeorgou2019} and
$\pi_2 = \pi_3 = 0$ recovers naive propensity score matching, so iDAPS
strictly generalises both. Rather than tuning $\bm\pi$ by hand
\citep[DAPS is sensitive to its tuning weight;][]{papadogeorgou2019}, the
weights are selected by minimising a balance criterion over the matched
sample,
\begin{equation}\label{eq:balance}
  \hat{\bm\pi} = \arg\min_{\bm\pi}\; B\{M(\bm\pi)\}, \qquad
  B\{M\} = \sum_{k=1}^{p}
    \Big|\frac{\bar X^{(T)}_k - \bar X^{(C)}_k}{s_k}\Big|
    \;+\; \bar d \;+\;
    \Big|\frac{\bar E^{(T)} - \bar E^{(C)}}{s_E}\Big| ,
\end{equation}
where $M(\bm\pi)$ is the 1:1 caliper-matched sample under $D^{\mathrm{comp}}_{ij}$, the first
and third terms are absolute standardised mean differences in covariates and
exposure between matched treated and controls ($s_k$ and $s_E$ are the
standard deviations of $X_k$ and $E$ among the matched treated units), and
$\bar d$ is the mean matched spatial distance. (The criterion mixes standardised and raw scales;
Proposition~\ref{prop:bound} supplies the principled rescaling spatial
distance enters the bias through $\sqrt{2\gamma_U(\cdot)}$ and the
software offers that variant, though we retain the simple form
\eqref{eq:balance} for continuity with DAPS.) The estimate is
$\widehat{\mathrm{ADET}}^{\mathrm{iDAPS}} = m_T^{-1}\sum_{i:A_i=1}(Y_i - Y_{j(i)})$
over the $m_T$ matched pairs, where $j(i)$ denotes the control matched to
treated unit $i$. Uncertainty for the matching estimators is
quantified by the classical matched-pairs standard error,
$\widehat{\mathrm{SE}} = \mathrm{sd}\{Y_i - Y_{j(i)}\}/\sqrt{m_T}$, with
normal critical values. Matched-pair differencing partially cancels the
spatial dependence; nearby pairs share the confounder surface, which
drops out of the difference; which is what makes a paired standard error
defensible here; it does not account for the estimated propensity score or
the weight search, for which the block bootstrap of
Section~\ref{sec:inference} is available.

\begin{algorithm}[h]
\caption{iDAPS: composite-distance matching with balance-optimised weights}
\label{alg:idaps}
\begin{algorithmic}[1]
\Require data $\{(Y_i, A_i, X_i, s_i)\}_{i=1}^{n}$; kernel bandwidth $\tau$;
caliper $c$; grid step $\delta$
\State fit the working propensity score $\hat p(X)$ by logistic regression
\State compute exposures $E_i \gets \sum_{j\ne i} G_{ij}(\tau) A_j$ as in
\eqref{eq:exposure}
\State form the pairwise distance matrices $D^{PS}$, $D^{\mathrm{Sp}}$,
$D^{\mathrm{Int}}$ and rescale each to $[0,1]$
\For{each $(\pi_1, \pi_2)$ on the $\delta$-grid of the simplex, with
$\pi_3 = 1 - \pi_1 - \pi_2$}
  \State $D^{\mathrm{comp}} \gets \pi_1 D^{PS} + \pi_2 D^{\mathrm{Sp}}
          + \pi_3 D^{\mathrm{Int}}$
  \State $M(\bm\pi) \gets$ 1:1 nearest-neighbour match of treated to
  controls, without replacement, within caliper $c$
  \State $B(\bm\pi) \gets$ balance criterion \eqref{eq:balance} on
  $M(\bm\pi)$
\EndFor
\State $\hat{\bm\pi} \gets \arg\min_{\bm\pi} B(\bm\pi)$;\quad
$M \gets M(\hat{\bm\pi})$ with matched pairs $\{(i, j(i))\}_{i=1}^{m_T}$
\State $\widehat{\mathrm{ADET}} \gets m_T^{-1} \sum_{i : A_i = 1}
\big(Y_i - Y_{j(i)}\big)$;\quad
$\widehat{\mathrm{SE}} \gets \mathrm{sd}\{Y_i - Y_{j(i)}\} / \sqrt{m_T}$
\Ensure $\widehat{\mathrm{ADET}}$, $\widehat{\mathrm{SE}}$, weights
$\hat{\bm\pi}$, matched pairs $M$
\end{algorithmic}
\end{algorithm}

\noindent\emph{Computational cost.} The distance matrices cost $O(n^2)$;
each grid point costs $O(n_T n_C)$ (with $n_T$, $n_C$ the numbers of treated
and control units) for the matching, giving $O(\delta^{-2} n_T n_C)$ overall.
In practice this runs in seconds for $n = 500$ on one core.

\subsection{A variogram bias bound}

A matched-pair difference is a \emph{spatial differencing} operation:
subtracting a nearby control cancels the shared confounder surface, leaving a
residual governed by how much $U$ varies over the matched distance. Let
$\gamma_U(h) = \tfrac12 \mathbb{E}\{U(s) - U(s+h)\}^2$ denote the semivariogram of
$U$ and $d_{ij(i)} = \lVert s_i - s_{j(i)} \rVert$ the matched distance.

\begin{proposition}[Bias decomposition and bound]\label{prop:bound}
Under \eqref{eq:outcome},
\begin{equation}\label{eq:decomp-bias}
  \mathbb{E}\big[\widehat{\mathrm{ADET}}^{\mathrm{iDAPS}}\big] - \theta_1
  \;=\; \theta_2^{\top} \mathbb{E}[\Delta X] + \theta_3 \mathbb{E}[\Delta E]
        + \theta_U \mathbb{E}[\Delta U],
\end{equation}
where $\Delta X, \Delta E, \Delta U$ are within-pair differences averaged
over matched pairs. Consequently
\begin{equation}\label{eq:bound}
  \big|\mathbb{E}[\widehat{\mathrm{ADET}}^{\mathrm{iDAPS}}] - \theta_1\big|
  \;\le\; \sum_{k=1}^{p} |\theta_{2,k}|\, \mathbb{E}|\Delta X_k|
        + |\theta_3|\, \mathbb{E}|\Delta E|
        + c_{\mathrm{ov}}\, |\theta_U|\,
          \mathbb{E}\Big[\sqrt{2\gamma_U\big(d_{ij(i)}\big)}\Big],
\end{equation}
where $c_{\mathrm{ov}} \ge 1$ is a selection constant accounting for the
dependence of the matching map on $U$ through the treatment. (By
Cauchy-Schwarz, the first term is further bounded by
$\lVert\theta_2\rVert \mathbb{E}\lVert\Delta X\rVert$; the coordinatewise form is
tighter and invariant to the scaling of individual covariates.)
\end{proposition}

A complete proof is given in Appendix~\ref{app:bound}, together with the
proof of Corollary~\ref{cor:cons}.

The bound reads as a division of responsibilities. The
first two terms are controlled directly by balancing covariates and exposure in the
match against an \emph{unmeasured} confounder, however, the only available
lever is the geometry of the match: how far apart matched units are, 
because the variogram $\gamma_U(\cdot)$ is a property of nature, not of the
design. Everything the matching estimator can do about spatial confounding,
it does by making $d_{ij(i)}$ small.

\begin{corollary}[Consistency and an irreducible floor]\label{cor:cons}
If the caliper $c_n \to 0$ slowly enough that $m_T \to \infty$ and $U$ is
mean-square continuous with $\gamma_U(0^+) = 0$ (no nugget), the confounding
term in \eqref{eq:bound} vanishes: matching \emph{nearby} is what removes
spatial confounding. If $U$ has a micro-scale nugget $\eta > 0$, an
irreducible floor $c_{\mathrm{ov}}|\theta_U|\sqrt{2\eta}$ remains, a limit
shared by every distance-based adjustment, including DAPS.
\end{corollary}

\begin{remark}[The bound justifies the composite metric]\label{rem:justify}
The three terms of \eqref{eq:bound} are exactly the three components of
$D^{\mathrm{comp}}$ in \eqref{eq:comp}, and the balance criterion \eqref{eq:balance} is
(up to standardisation) an empirical estimate of the bound. The data-driven
weight selection is therefore \emph{bound minimisation}, not an ad hoc
heuristic. It also explains the failure mode of DAPS under interference:
DAPS controls the first and second distances but leaves $\mathbb{E}|\Delta E|$
uncontrolled, so its bound carries a free interference term.
\end{remark}

All quantities in \eqref{eq:bound} are estimable: the coefficients from a
working outcome regression, and $\gamma_U$ from the fitted spatial variogram
of its residuals. The resulting plug-in diagnostic is reported for the
application in Section~\ref{sec:application}.

Both the decomposition and the bound are validated numerically in
Section~\ref{sec:sim-bound}, together with a caliper sweep that traces out
Corollary~\ref{cor:cons} directly (Figure~\ref{fig:t2}). A companion design
property; because the weights $\hat{\bm\pi}$ optimise balance, they
adapt to the strength of unmeasured confounding, abandoning the propensity
score exactly when it becomes unreliable is demonstrated in
Section~\ref{sec:weights}.

\section{recoverU$+$: doubly robust estimation with a recovered confounder}
\label{sec:recoveru}

\subsection{The estimator}

RecoverU$+$ proceeds in two stages, which have a simple reading. The first
stage draws a \emph{map of the missing confounder}: after removing what
treatment, covariates and exposure can explain, the remaining smooth
spatial variation in the outcome is attributed to $U$ and reconstructed by
kriging. This reconstruction is precisely the recoverable component
$U_{R}$ of \eqref{eq:decomp}. The second stage then treats the reconstruction
as one more covariate in an otherwise standard doubly robust analysis.

With $V = (X, E, \widehat U_{R})$, the estimator is the ATT-type doubly robust
form of \citet{moodie2018doubly},
\begin{equation}\label{eq:dr}
  \hat\tau \;=\; \frac{1}{n_1} \sum_{i=1}^{n}
     \Big[ A_i - (1 - A_i)\, w(V_i) \Big]\big( Y_i - \mu_0(V_i) \big),
  \qquad w = \frac{p}{1-p},\ \ n_1 = \textstyle\sum_i A_i ,
\end{equation}
where $p(V)$ is a working model for the adjustment probability
$\mathbb{P}(A = 1 \mid V)$ of Remark~\ref{rem:ps}, fitted by logistic
regression; $w = p/(1-p)$ is the corresponding odds; and $\mu_0(V)$ is a working model
for the control-outcome mean $\mathbb{E}[Y \mid A = 0, V]$, fitted on the control
units. The full pipeline is Algorithm~\ref{alg:recoveru}. Note the
contrast with the matching estimator: iDAPS fits its propensity score on
$X$ alone and handles exposure and location through the composite
distance, whereas here the propensity model conditions on $X$, $E$
\emph{and} $\widehat U_{R}$ directly, it is the estimator that
operationalises the identification result of
Proposition~\ref{prop:id} literally.

\begin{algorithm}[h]
\caption{recoverU$+$: doubly robust estimation with a recovered confounder}
\label{alg:recoveru}
\begin{algorithmic}[1]
\Require data $\{(Y_i, A_i, X_i, s_i)\}_{i=1}^{n}$; kernel bandwidth $\tau$
\State compute exposures $E_i \gets \sum_{j\ne i} G_{ij}(\tau) A_j$ as in
\eqref{eq:exposure}
\Statex \textit{Stage 1: recovery of the spatial confounder}
\State fit the initial outcome model $Y \sim A + X + E$ by least squares;
store residuals $r_0$
\State fit a Mat\'ern covariance (smoothness $\nu = 1/2$) with nugget to
$r_0$ by maximum likelihood, yielding the partial sill $\hat\sigma^2$,
range $\hat\theta$ and nugget $\hat\sigma^2_\epsilon$, and hence the
implied $n \times n$ covariance matrix $\widehat\Sigma_U$
\State refit the initial model by generalised least squares under
$\widehat\Sigma_U + \hat\sigma^2_\epsilon I$; store residuals $\hat r$
\State $\widehat{U_{R}} \gets \widehat\Sigma_U\,
(\widehat\Sigma_U + \hat\sigma^2_\epsilon I)^{-1}\, \hat r$
\Comment{kriging projection \citep{pokal2023improved}}
\Statex \textit{Stage 2: doubly robust estimation}
\State $V \gets (X, E, \widehat U_{R})$; fit the logistic propensity model
$p(V)$ and the control-outcome model $\mu_0(V)$ on the $A = 0$ units
\State $\hat\tau \gets$ \eqref{eq:dr} with weights $w = p/(1-p)$
\State $\widehat{\mathrm{SE}} \gets$ spatial HAC on the influence values
$\hat\psi_i$, or spatial block bootstrap (Section~\ref{sec:inference})
\Ensure $\hat\tau$, $\widehat{\mathrm{SE}}$, recovered confounder
$\widehat{U_{R}}$
\end{algorithmic}
\end{algorithm}

\noindent\emph{Computational cost.} Each Mat\'ern likelihood evaluation is
$O(n^3)$ (dense Cholesky); the recovery stage dominates the runtime, at
roughly $2$\,s for $n = 150$ and $15$\,s for $n = 300$ per fit on one core.
A practical recommendation from our experiments: fix the Mat\'ern smoothness
at $\nu = 1/2$ (exponential), as in Algorithm~\ref{alg:recoveru}, the
free four-parameter likelihood is poorly identified on weak residual fields
and can drive $\hat\nu$ to a numerical boundary (we observed
$\hat\nu > 100$).

recoverU differs only by omitting $E$ from $V$; it adjusts for confounding
but not interference. Relative to \citet{pokal2023improved}, the innovation
of recoverU$+$ is that both nuisance models condition on the exposure, so SC
and SI are handled simultaneously in a single estimator.

\subsection{Partial double robustness}

The classical justification of \eqref{eq:dr} is double robustness:
consistency if either $p$ or $\mu_0$ is correctly specified
\citep{funk2011doubly}. With a recovered confounder this \emph{cannot} hold
without qualification, because the adjustment set contains $U_{R}$, not $U$.
The correct statement is the following.

\begin{proposition}[Partial double robustness]\label{prop:pdr}
Assume \eqref{eq:outcome}, so that
$Y(0) = \mu_0(V) + \theta_U U_{A} + \tilde\varepsilon$, where
$\mu_0(V) = \theta_2^{\top} X + \theta_3 E + \theta_U U_{R}$ is the
$V$-measurable part of $Y(0)$ and
$\mathbb{E}[\tilde\varepsilon \mid V, A] = 0$, and suppose at least one of the two
working models $\{p, \mu_0\}$ is correctly specified given
$V = (X, E, U_{R})$. Then
\begin{equation}\label{eq:bua}
  \mathbb{E}[\hat\tau] - \mathrm{ADET}
  \;=\; \theta_U \Big( \mathbb{E}[U_{A} \mid A{=}1] - \mathbb{E}^{w}[U_{A} \mid A{=}0] \Big)
  \;=:\; b_{U_{A}},
  \qquad
  \mathbb{E}^{w}[\,\cdot \mid A{=}0] := \frac{\mathbb{E}[(1-A) w \,\cdot\,]}{\mathbb{E}[(1-A) w]} .
\end{equation}
In particular: (i) if $U \perp \!\!\! \perp A$ then $U_{A} \equiv 0$ and $b_{U_{A}} = 0$,
classical double robustness is recovered; and (ii)
$|b_{U_{A}}| \le |\theta_U|\, \sigma_{U_{A}}\, \kappa$, where
$\sigma_{U_{A}}^2 = \mathrm{Var}(U_{A})$ and $\kappa$ is an overlap constant
depending only on the treatment prevalence and the tail of the odds
weights (its explicit form is given in Appendix~\ref{app:pdr}).
\end{proposition}

A complete proof, with the population algebra written out term by term, is
given in Appendix~\ref{app:pdr}.

Three readings. First, this is the more precise replacement for an unqualified
``doubly robust'' label: robustness holds \emph{with respect to the feasible
adjustment set}, up to an explicit residual. Second, $b_{U_{A}}$ is
simultaneously the identification gap of Proposition~\ref{prop:id}(b) and the
asymptotic bias of $\hat\tau$. Third, under a probit treatment model
$b_{U_{A}}$ has a closed form in the correlation and scale of $U_{A}$ (by
standard inverse-Mills-ratio algebra for a jointly Gaussian confounder and
latent treatment index), anchoring a sensitivity analysis in the spirit of
the E-value \citep{vanderweele2017}; developing that sensitivity analysis
is beyond the scope of this paper and is left to future work.

Proposition~\ref{prop:pdr} is validated numerically in
Section~\ref{sec:sim-pdr}, including a direct check of claim~(i) and an
oracle comparison isolating the unrecoverable component
(Figure~\ref{fig:t3}).

\subsection{Inference under spatial dependence}\label{sec:inference}

Estimator \eqref{eq:dr} is asymptotically linear with influence values
$\hat\psi_i = [A_i - (1-A_i)\hat w_i](Y_i - \hat\mu_0(V_i))/\hat\pi_1 -
A_i \hat\tau/\hat\pi_1$, where $\hat\pi_1 = n_1/n$ is the treated
fraction. The conventional standard error
$\mathrm{sd}(\hat\psi_i)/\sqrt{n}$ treats the $\hat\psi_i$ as independent;
in spatial data they are positively correlated and the resulting intervals
are too narrow. We implement three remedies in \textsf{spaci}.

\emph{Spatial HAC (heteroskedasticity- and autocorrelation-consistent)
variance} \citep{conley1999}:
$\widehat{\mathrm{Var}}(\hat\tau) = n^{-2} \sum_{i}\sum_{j}
K(\lVert s_i - s_j\rVert / b)\, \hat\psi_i \hat\psi_j$, with a Bartlett
kernel and bandwidth $b$ chosen from the distance at which the empirical
autocorrelation of the $\hat\psi_i$ decays below $0.1$.

\emph{Spatial block bootstrap:} square spatial blocks (side
$\propto \mathrm{diam}(D)\, n^{-1/4}$) resampled with replacement. For the
doubly robust estimators, the bootstrap resamples blocks of \emph{influence
values}; re-running the full pipeline on resampled coordinates is unsound
here, because resampling duplicates locations and the kriging recovery
over-smooths them, spuriously \emph{shrinking} the variance a problem
we document so that others avoid it. For the matching estimators, which have
no kriging stage, the pipeline is re-run on each resample with exposures
recomputed.

\emph{Conditional randomisation test} for the sharp null of no direct
effect: treatment is redrawn from the fitted propensity model, exposures
recomputed, and the statistic re-evaluated \citep{basse2019}; a design-based
$p$-value that does not rest on variance estimation.

The finite-sample calibration of all three procedures is assessed by
simulation in Section~\ref{sec:sim-calib}. Two findings from that study
frame everything that follows. In the regime the methods target
meaningful spatial confounding, the independence-based standard error is
anticonservative while the spatial HAC and block bootstrap are
approximately calibrated to conservative, so the practical recommendation
is unambiguous: \emph{report HAC or block-bootstrap intervals, and avoid the
i.i.d.\ standard error}. And when confounding is \emph{weak}, all
influence-based standard error, including the three above -
undercover, for a structural reason that no functional of the influence values can repair; the mechanism, i.e a generated-regressor problem and
the most practical correction is set out in Section~\ref{sec:discussion}.

\section{Simulation studies}\label{sec:sims}

\subsection{Software and simulation design}\label{sec:software}

All of the methodology above is implemented in the open-source R package
\textsf{spaci} (\url{https://github.com/Ogunsolaia/spaci}), which provides the
five estimators (\texttt{iDAPS()}, \texttt{recoverUplus()} and the
comparators \texttt{naive\_PS()}, \texttt{DAPS()}, \texttt{recoverU()}),
the exposure and recovery machinery, the spatial inference procedures of
Section~\ref{sec:inference} (\texttt{vcov\_hac()}, \texttt{boot\_spatial()},
\texttt{rand\_test()}), the bias-bound diagnostic of
Proposition~\ref{prop:bound} (\texttt{bias\_bound()}), deterministic optimal
matching as an alternative to seeded greedy matching, the fixed-$\nu$
recovery option of Algorithm~\ref{alg:recoveru}, the simulator used
throughout this section, and plotting utilities. All estimators share a
common interface, so a complete analysis is a few lines:

\begin{verbatim}
R> library(spaci)
R> fit <- recoverUplus(Y, A, X, coords, tau = 0.2)   # doubly robust
R> vcov_hac(fit)                                     # spatial HAC interval
R> m <- idaps(Y, A, X, coords, tau = 0.2,            # matching
+             caliper = 0.25, seed = 1)
R> bias_bound(m, Y, A, X, coords)                    # Prop. 2 diagnostic
R> boot_spatial(Y, A, X, coords, method = "idaps")   # block bootstrap
\end{verbatim}

\noindent The package passes \texttt{R CMD check} cleanly and ships the
application data with documentation; every experiment in this section, and
the application of Section~\ref{sec:application}, runs through it, and a
separate public reproduction repository provides the scripts that
regenerate every table and figure in this paper (see \emph{Code
availability}). Computational profiles were given alongside
Algorithms~\ref{alg:idaps} and~\ref{alg:recoveru}, and scalability options
for large $n$ are discussed in Section~\ref{sec:discussion}.

The common data-generating process is as follows. Units are placed
uniformly on $[0,1]^2$ with $n = 250$. Covariates
$X_{1} \sim N(0,1)$, $X_{2} \sim \mathrm{Bernoulli}(0.5)$; the unmeasured
confounder $U$ is a mean-zero Gaussian random field with exponential
covariance (variance $1$, range $0.2$); treatment follows
$A_i \sim \mathrm{Bernoulli}\big(\mathrm{logit}^{-1}\{0.1 + 0.1 X_{i1} +
0.2 X_{i2} + u\, U(s_i)\}\big)$, with $u$ the spatial-confounding strength;
exposure $E_i$ uses \eqref{eq:exposure} with $\tau = 0.1$; and outcomes
follow \eqref{eq:outcome} with
$(\theta_1, \theta_{2i}, \theta_3, \theta_U), = (2.0, (1.0, 0.5), 1.5, 0.4)$ where $i=1,2 $, so
$\mathrm{ADET} = 2$. Each configuration uses $1{,}000$ replicates. Five estimators
are compared: naive propensity-score (PS) matching, DAPS, iDAPS, recoverU and
recoverU$+$. (The proposition-validation studies of
Sections~\ref{sec:sim-bound} and~\ref{sec:sim-pdr} use the same simulator
with confounder loadings up to $\theta_U = 2$ so that the recovery stage is
well identified across the grid; full configurations are in the
reproducibility materials.) The confounder loading $\theta_U = 0.4$ keeps
the residual confounder field deliberately weak. If anything, a handicap
for the recovery-based estimators. The validation studies increase it
so that the recovery stage has signal to work with an honest
prerequisite, since with a weak residual field the Mat\'ern fit is barely
identified (Section~\ref{sec:recoveru}). One comparability remark: the naive method
ignores exposure entirely and so targets a quantity that mixes direct and
spillover components; comparing all five methods against $\mathrm{ADET} = 2$ is
coherent here only because the direct effect is constant in the
data-generating process. In applications with heterogeneous effects the
naive method would differ in estimand, not merely in bias.

\subsection{Validating the bias bound (Proposition~\ref{prop:bound})}
\label{sec:sim-bound}

This and the next subsection directly validate the two headline guarantees; Propositions~\ref{prop:bound} and~\ref{prop:pdr} before the
head-to-head comparison and stress tests that follow. Both reuse the
geostatistical simulator above, continuous coordinates on $[0,1]^2$, and
an unmeasured confounder $U(s_i)$ drawn as a mean-zero Gaussian random field
with exponential covariance (range $0.2$, variance $1$), obtained by
Cholesky-factorising the $n\times n$ covariance matrix evaluated at the
sampled coordinates and multiplying by independent standard normals. So
every replicate is a fresh continuous spatial field, not a discretised
approximation to one.

Because the working outcome model~\eqref{eq:outcome} is linear and, in
simulation, the true coefficients and the true confounder $U$ are known, the
matched-pair bias decomposition~\eqref{eq:decomp-bias} can be evaluated
\emph{exactly} rather than merely estimated, which is what makes this a
validation of the proposition rather than an application of the diagnostic.
For each replicate: draw one data set, fit iDAPS (Algorithm~\ref{alg:idaps},
$\tau=0.1$, caliper $0.25$), and on the resulting matched pairs compute (a)
the right-hand side of~\eqref{eq:decomp-bias} using the true $\theta_2,
\theta_3,\theta_U$ and the true within-pair differences $\Delta X, \Delta E,
\Delta U$; (b) the true bound~\eqref{eq:bound}, using the \emph{population}
semivariogram $\gamma_U(h) = 1 - e^{-h/0.2}$ of the exponential field
evaluated at the matched distances; and (c) the \emph{plug-in} diagnostic
that a practitioner could compute from data alone; coefficients from a
fitted working regression and $\gamma_U$ from the fitted exponential
variogram of its residuals (the same computation as \texttt{bias\_bound()}
in \textsf{spaci}, used on the real data in
Section~\ref{sec:application}). Sweeping a $3\times3$ grid of confounder
loading $\theta_U \in \{0.4,1,2\}$ (the outcome-model coefficient on $U$)
and interference strength $\theta_3 \in \{0,1.5,3\}$, with the
treatment-model confounding strength held fixed at $u = 2$, gives $9$
design points;
$200$ replicates per design point at $n=250$ yield $1{,}800$ replicates in
total.

The results confirm every part of the proposition. The exact decomposition
held to Monte-Carlo error ($\le 0.02$ at every design point), confirming the
algebra of~\eqref{eq:decomp-bias} rather than just its consequence. The true
bound~\eqref{eq:bound} exceeded the realised $|$bias$|$ in \textbf{100\%} of
the $1{,}800$ replicates. The confounding term dominated the other two
(covariate and exposure imbalances were driven to $\approx 0.06$ by the
matching, so essentially all residual bias is unmeasured-confounder
imbalance);  the regime the bound is designed for. The selection constant
behaves exactly as the proof anticipates: \emph{on average} the confounding
term with $c_{\mathrm{ov}} = 1$ already dominates the confounder
contribution at all nine design points, but per replicate the matching map
correlates with $U$ (treated units carry systematically larger $U$), so
the realised within-pair differences fluctuate above the population
variogram term in roughly $30\%$ of individual replicates, with a
$95$th-percentile ratio of about $1.2$, which is the empirical origin used throughout the paper i.e. $c_{\mathrm{ov}} = 1.2$ . Finally, the
\emph{plug-in} diagnostic (c), the one available in
practice, with no knowledge of the truth, recovered $94\%$ of the true
confounding term, so the applied diagnostic reported in
Section~\ref{sec:application} is not merely a bound in principle but a
reasonably tight one in practice.

A separate caliper sweep isolates Corollary~\ref{cor:cons}: at fixed
$\theta_U = 2$, $\theta_3 = 1.5$, iDAPS is refitted at calipers
$\{0.05, 0.1, 0.15, 0.25, 0.5, \infty\}$, where $\infty$ means no caliper
--- every treated unit is matched to its nearest control, however distant
($200$ replicates each, $n=250$). Figure~\ref{fig:t2} (left) shows the
mechanism at work: as the caliper tightens and the mean matched distance
falls $0.24 \to 0.10$, the confounding bound term, the actual within-pair
confounder contribution and the realised bias all fall together, the
realised bias from $2.69$ to $1.40$ and the right panel confirms that
the total estimated bound dominates the realised bias at every caliper.
Matching nearby is what removes spatial confounding, exactly as the
consistency argument predicts. Two honest footnotes to the sweep: the
price of a tighter caliper is matched-sample size (from $117$ to $41$
matched treated units across the sweep), and the bias does not reach zero
at fixed $n$ because nearest-neighbour distances floor at $\approx 0.06$
when $n = 250$, the corollary's limit requires the matched distances to
shrink \emph{with} $n$, not at fixed $n$.

\begin{figure}[H]\centering
  \includegraphics[width=0.95\linewidth]{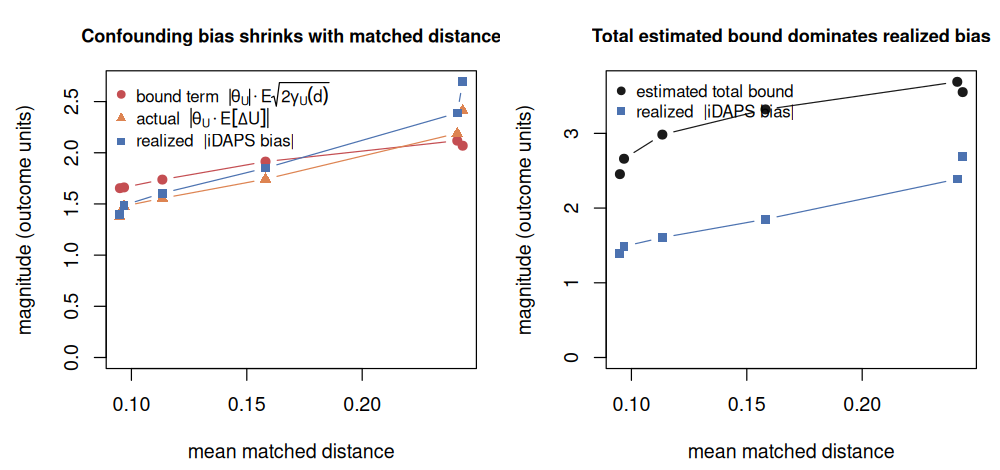}
  \caption{Validation of Proposition~\ref{prop:bound} and
  Corollary~\ref{cor:cons}. \emph{Left:} the confounding bound term, the
  actual within-pair confounder contribution, and the realised iDAPS bias
  all shrink together as units are matched more closely. \emph{Right:} the
  total estimated bound dominates the realised bias across all calipers.}
  \label{fig:t2}
\end{figure}

\subsection{Validating partial double robustness (Proposition~\ref{prop:pdr})}
\label{sec:sim-pdr}

The same simulator is used to isolate the two ingredients of
Proposition~\ref{prop:pdr}: the outcome-model confounder loading $\theta_U$
and, separately, the strength $u$ with which the confounder drives
\emph{treatment} (the $u\,U(s_i)$ term inside the treatment logit of
Section~\ref{sec:software}; $u = 0$ means $U \perp \!\!\! \perp A$ and classical double
robustness should be exactly recovered). Crossing
$\theta_U \in \{0.4,1,2\}$ with
$u \in \{0,0.5,1,2,3\}$ gives $15$ design points; $200$ replicates
per design point at $n=250$. Each replicate runs the full recoverU$+$
pipeline of
Algorithm~\ref{alg:recoveru} (initial linear fit, exponential Mat\'ern MLE
on its residuals, GLS refit, kriging projection to $\widehat{U_{R}}$, then the
doubly robust estimate~\eqref{eq:dr} with $V = (X, E, \widehat{U_{R}})$), and
records the realised bias $\hat\tau - \mathrm{ADET}$.

Because the simulation knows the true $U$, the analytic residual $b_{U_{A}}$
of~\eqref{eq:bua} can be computed alongside the estimate, by decomposing $U$
into the part linearly explained by the adjustment set $(X, E,
\widehat{U_{R}})$ and an orthogonal remainder $r_U$, the empirical analogue
of $U_{A}$ and evaluating
$b_{U_{A}} = \theta_U\big(\bar r_U^{\,\text{treated}} -
\bar r_U^{\,\text{odds-weighted control}}\big)$ exactly as
in~\eqref{eq:bua}. Across all $3{,}000$ replicates the realised bias tracked
this analytic $b_{U_{A}}$ with per-replicate correlation $0.94$ (regression
slope $0.90$ for the design-point means against the $45^\circ$ identity),
directly
confirming the formula rather than just its qualitative implications. Claim
(i) is checked by holding $\theta_U$ fixed and driving $u \to 0$: as
the confounder ceases to influence treatment, $U_{A} \to 0$ and the bias
should vanish, which it does (mean $|\mathrm{bias}| = 0.011$ at the
$u = 0$ design points, statistically indistinguishable from zero).
Finally, a
model-free check that the residual bias is genuinely the \emph{unrecoverable}
component and not, say, an artefact of the kriging step: rerunning the
identical pipeline with the \emph{true, full} $U(s)$ substituted for
$\widehat{U_{R}}$ (an oracle no real analysis has access to) is unbiased at
every one of the $15$ design points (mean $|\mathrm{bias}| = 0.019$, at
Monte-Carlo
noise level), regardless of $u$ or $\theta_U$, confirming that
once the adjustment set genuinely spans $U$, ordinary double robustness
holds without qualification, and that the bias observed with the recovered
$\widehat{U_{R}}$ traces entirely to the $U_{A}$ that recovery cannot see
(Figure~\ref{fig:t3}). (The two summaries, $0.011$ and $0.019$, average
over different subsets, the three $u = 0$ design points versus
all fifteen and both sit within Monte-Carlo error of zero, so their
ordering carries no information.)

\begin{figure}[H]\centering
  \includegraphics[width=0.95\linewidth]{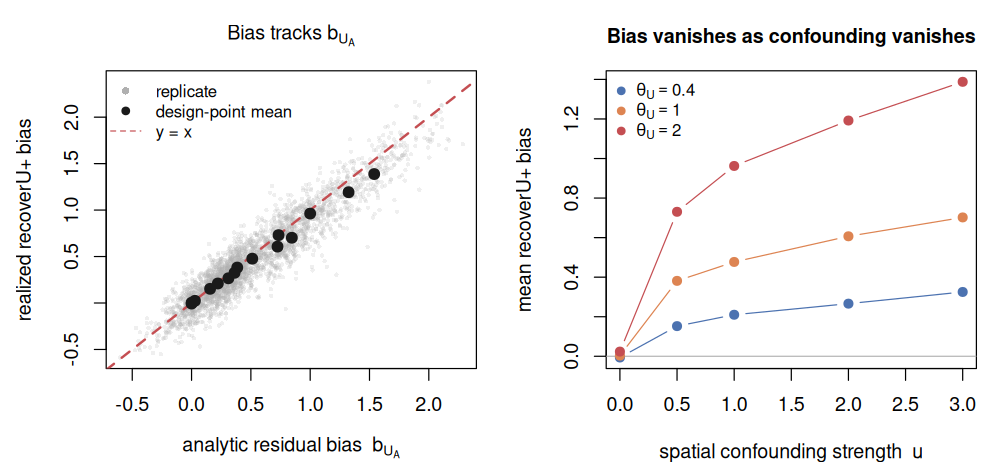}
  \caption{Validation of Proposition~\ref{prop:pdr}. \emph{Left:} realised
  recoverU$+$ bias against the analytic residual $b_{U_{A}}$; grey points are
  replicates, black points design-point means, dashed line the identity.
  \emph{Right:} bias vanishes as the confounder ceases to drive treatment
  and scales with the confounder's outcome loading $\theta_U$.}
  \label{fig:t3}
\end{figure}

\subsection{The weights adapt to confounding strength}\label{sec:weights}

The next experiment probes the balance-optimised weights of
Section~\ref{sec:idaps} themselves, rather than the resulting estimate.
Crossing the treatment-model confounding strength $u \in \{0, 0.5, 1, 2\}$ with
interference strength $\theta_3 \in \{0, 1.5, 3\}$ ($160$ replicates per
design point, $n = 250$), each replicate fits iDAPS and records the
selected weights $\hat{\bm\pi}$; Table~\ref{tab:weights} reports their
means by $u$, averaged over the interference settings.

Because $\hat{\bm\pi}$ optimises balance (a design quantity), it responds
to the amount of unmeasured confounding that must be absorbed, not to the
outcome model. As $u$ increases, the
propensity-score weight $\pi_1$ \emph{collapses} (correlation $-0.92$ with
$u$) while the spatial and exposure weights take over ($\pi_3$ correlation
$+1.00$); the weights are essentially invariant to the interference
\emph{outcome} coefficient (correlation $\approx 0$), as expected of a
design-based rule. iDAPS automatically stops trusting the propensity score built from the \emph{measured} covariates precisely when unmeasured
confounding renders it unreliable.

\begin{table}[h]
\centering
\caption{Mean optimised iDAPS weights by spatial confounding strength $u$
(averaged over interference strengths; $160$ replicates per design point,
$n=250$).}
\label{tab:weights}
\begin{tabular}{lccc}
\toprule
 & $\pi_1$ (propensity) & $\pi_2$ (spatial) & $\pi_3$ (exposure) \\
\midrule
$u = 0$ & 0.42 & 0.42 & 0.16 \\
$u = 0.5$ & 0.31 & 0.53 & 0.17 \\
$u = 1$ & 0.13 & 0.54 & 0.33 \\
$u = 2$ & 0.08 & 0.38 & 0.54 \\
\bottomrule
\end{tabular}
\end{table}

\subsection{Calibration of the inference procedures}\label{sec:sim-calib}

This experiment asks whether the standard errors of
Section~\ref{sec:inference} are the right size. Data are drawn from the
same simulator with the treatment-model confounding strength swept over
$u \in \{0, 0.5, 1, 2\}$ at $n = 150$, plus one $u = 0$ configuration at
$n = 300$ to separate small-sample from structural effects. recoverU$+$ is
fitted to each of $300$--$400$ replicates per configuration, and each
replicate records the i.i.d.\ influence-function standard error, the
spatial HAC standard error, and the spatial block-bootstrap standard error
($400$ resamples). Calibration is then summarised in two ways in
Table~\ref{tab:coverage}: the ratio of the mean estimated standard error to
the true Monte-Carlo standard deviation of the estimator across replicates
(ratio $1$ = correctly sized), and the empirical coverage of nominal
$95\%$ intervals. Coverage is assessed for the estimator's
\emph{probability limit} rather than for the true $\mathrm{ADET}$ in the
confounded configurations: Proposition~\ref{prop:pdr} already characterises the bias
there, and an interval can only be expected to cover what the estimator
converges to. Conflating the two would make every interval look bad for
a reason that has nothing to do with its width being right, which is the
question this experiment isolates.

In the regime the method targets meaningful spatial confounding,
the independence-based standard error remains anticonservative (ratios
$0.85$--$0.93$) while the spatial HAC and block bootstrap move from
approximately calibrated to conservative (ratios $1.26$ and $1.31$,
coverage $0.98$, at the strongest confounding). This is the basis for the
recommendation of Section~\ref{sec:inference}: report HAC or
block-bootstrap intervals, and avoid the i.i.d.\ standard error.

The weak-confounding rows ($u = 0$) reveal a genuine failure: \emph{all}
influence-based standard errors undercover (SE ratios $\approx 0.55$,
coverage $\approx 0.75$), and the shortfall does not shrink from $n = 150$
to $n = 300$, the signature of a missing variance component rather than
a small-sample artefact. We do not hide this: its mechanism and the most
practical correction are analysed in Section~\ref{sec:discussion}.

\begin{table}[t]
\centering
\caption{Standard-error calibration and coverage for recoverU$+$ across
spatial confounding strength $u$ ($n = 150$ unless stated; $300$--$400$
replicates per configuration; block bootstrap with $400$ resamples).
Ratio $=$ mean estimated SE $\div$ true Monte-Carlo SD; coverage is for
the probability limit at nominal $0.95$.}
\label{tab:coverage}
\begin{tabular}{lcccccc}
\toprule
 & \multicolumn{3}{c}{SE ratio} & \multicolumn{3}{c}{Coverage} \\
\cmidrule(lr){2-4}\cmidrule(lr){5-7}
$u$ & i.i.d. & HAC & bootstrap & i.i.d. & HAC & bootstrap \\
\midrule
0.0 & 0.58 & 0.57 & 0.54 & 0.75 & 0.74 & 0.68 \\
0.5 & 0.85 & 0.86 & 0.86 & 0.89 & 0.89 & 0.88 \\
1.0 & 0.93 & 1.04 & 1.12 & 0.93 & 0.93 & 0.96 \\
2.0 & 0.91 & 1.26 & 1.31 & 0.93 & 0.98 & 0.98 \\
\midrule
0.0 ($n{=}300$) & 0.60 & 0.60 & 0.57 & 0.79 & 0.79 & 0.75 \\
\bottomrule
\end{tabular}
\end{table}

\subsection{Comparison of the five estimators}
\label{sec:sim-headtohead}

The four preceding experiments each isolated a single theoretical claim
under conditions chosen to test it. The remaining two ask the practical
questions. This one is the headline comparison: all five estimators;
naive propensity-score matching, DAPS, iDAPS, recoverU and recoverU$+$
are run on the \emph{same} simulated data sets, drawn from the
data-generating process of Section~\ref{sec:software} exactly as stated
there (outcome confounder loading $\theta_U = 0.4$, the deliberately weak
regime that, if anything, handicaps the recovery-based estimators), with
the treatment-model confounding strength swept
over $u \in \{0.5, 1, 1.5, 2\}$ and $1{,}000$ replicates per
configuration. It answers the question an applied reader ultimately cares
about: when spatial confounding and interference operate together, how
much accuracy does each successive refinement buy?

Table~\ref{tab:sim} reports the bias and mean squared error in parentheses. Three
patterns are stable. First,
all methods improve as confounding weakens. Second, the two proposed
estimators dominate their single-problem antecedents: iDAPS improves on DAPS
exactly as Remark~\ref{rem:justify} predicts, and recoverU$+$ improves on
recoverU for the same structural reason. Third, recoverU$+$ has the lowest MSE across the simulated confounding strengths considered. At the strongest confounding, its MSE ($0.109$) is roughly a
fifth of the naive method's ($0.566$) and its remaining bias is the
$b_{U_{A}}$ of Proposition~\ref{prop:pdr}, as validated in
Figure~\ref{fig:t3}. The methods rank
recoverU$+$ $\succ$ iDAPS $\succ$ recoverU $\succ$ DAPS $\succ$ naive PS.

\begin{table}[h]
\centering
\caption{Head-to-head comparison of the five estimators across spatial
confounding strength $u$, ($n = 250$, $1{,}000$ replicates; true
$\mathrm{ADET} = 2$, so the mean estimate is $2 + $ bias).}
\label{tab:sim}
\begin{tabular}{lcccc}
\toprule
Method & $u=2.0$ & $u=1.5$ & $u=1.0$ & $u=0.5$ \\
\midrule
Naive PS
& 0.733 (0.566)
& 0.633 (0.430)
& 0.467 (0.247)
& 0.263 (0.098) \\

DAPS
& 0.693 (0.508)
& 0.591 (0.376)
& 0.433 (0.215)
& 0.240 (0.085) \\

iDAPS
& 0.445 (0.232)
& 0.406 (0.198)
& 0.327 (0.138)
& 0.214 (0.073) \\
recoverU
& 0.694 (0.507)
& 0.592 (0.374)
& 0.422 (0.202)
& 0.217 (0.069) \\

recoverU$+$
& \textbf{0.289 (0.109)}
& \textbf{0.269 (0.094)}
& \textbf{0.220 (0.070)}
& \textbf{0.161 (0.045)} \\
\bottomrule
\end{tabular}
\end{table}

\subsection{Double robustness and exposure misspecification}
\label{sec:misspec}

Two further experiments probe the framework where it could fail.

\emph{Double robustness.} We add a shared nonlinearity, the
treatment model gains $0.45(X_1^2 - 1)$ inside the logit and the outcome
gains $0.8 X_1^2$ and set $u = 0$ so that the unrecoverable component
vanishes and Proposition~\ref{prop:pdr} predicts \emph{exact} double
robustness: unbiasedness whenever at least one nuisance specification
includes the $X_1^2$ term, bias only when both omit it. ``Outcome model
correct'' means the outcome specification used both in the
recovery stage and as $\mu_0$ includes $X_1^2$; ``propensity model
correct'' means the logistic model includes it. Table~\ref{tab:dr} reports
the four cells ($n = 250$, $150$ replicates). The prediction is confirmed
with one instructive nuance: when both models are correct or only the
propensity model is misspecified, the estimator is unbiased ($|$bias$|
\le 0.017$); when both are misspecified, it is clearly biased ($0.55$); and
when only the \emph{outcome} model is misspecified the estimator is
asymptotically protected by the correct propensity model but finite-sample
fragile. The mean bias of $0.23$ comes with a tripled standard deviation
($0.71$), the signature of extreme inverse-probability weights when
protection must come from weighting alone. This is precisely the phenomenon
documented by \citet{kang2007}, reproduced here in the spatial setting, and
it reinforces the practical advice that the outcome model deserves the
greater specification effort.

\begin{table}[h]
\centering
\caption{Double-robustness stress test for recoverU$+$ ($u = 0$, $n = 250$,
$150$ replicates; true $\mathrm{ADET} = 2$). A shared nonlinearity ($X_1^2$) drives
both treatment and outcome; each nuisance specification either includes it
(correct, C) or omits it (misspecified, M).}
\label{tab:dr}
\begin{tabular}{lcccc}
\toprule
Cell & Propensity model & Outcome model & Bias & SD \\
\midrule
CC & correct & correct & $0.017$ & $0.23$ \\
MC & misspecified & correct & $0.011$ & $0.22$ \\
CM & correct & misspecified & $0.226$ & $0.71$ \\
MM & misspecified & misspecified & $0.545$ & $0.27$ \\
\bottomrule
\end{tabular}
\end{table}

\emph{Exposure-mapping misspecification.} The estimators always use the
exponential kernel \eqref{eq:exposure} with $\tau = 0.1$, while the
data-generating exposure is (a) the same (correct), (b) an exponential
kernel with $\tau = 0.05$ (bandwidth wrong), or (c) the fraction of treated
units among the five nearest neighbours (functional form wrong), at
$u = 0$ and interference strength $1.5$. Both recoverU$+$ and iDAPS are
essentially unbiased in all three cases ($|$bias$| \le 0.016$ throughout;
Table~\ref{tab:expo}): because the working exposure remains highly
correlated with the true one, misspecification mainly inflates residual
noise rather than biasing the direct effect, consistent with the
robustness perspective of \citet{savje2024}. We caution that this
conclusion is specific to direct effects; spillover estimands, which live
on the exposure margin itself, would be more sensitive.

\begin{table}[t]
\centering
\caption{Exposure-mapping misspecification ($u = 0$, interference strength
$1.5$, $n = 250$, $150$ replicates; true $\mathrm{ADET} = 2$). Estimators always use
the exponential kernel with $\tau = 0.1$; the data-generating exposure
varies by row.}
\label{tab:expo}
\begin{tabular}{lcccc}
\toprule
 & \multicolumn{2}{c}{recoverU$+$} & \multicolumn{2}{c}{iDAPS} \\
\cmidrule(lr){2-3}\cmidrule(lr){4-5}
Data-generating exposure & Bias & SD & Bias & SD \\
\midrule
Exponential, $\tau = 0.1$ (correct) & $-0.015$ & $0.25$ & $0.004$ & $0.26$ \\
Exponential, $\tau = 0.05$          & $\phantom{-}0.008$ & $0.25$ & $0.004$ & $0.26$ \\
Five-nearest-neighbour fraction     & $-0.006$ & $0.25$ & $0.006$ & $0.27$ \\
\bottomrule
\end{tabular}
\end{table}

\section{Effect of SCR/SNCR technology on ambient ozone}
\label{sec:application}

Ambient ozone is one of the air pollutants most associated with human
disease. To reduce ozone formation resulting from nitrogen-oxide
(NO$_x$) emissions, selective catalytic and non-catalytic reduction
(SCR/SNCR) technologies have been widely adopted at power-generating
facilities. Whether SCR/SNCR use reduces ambient ozone has been investigated
under spatial confounding alone
\citep{papadogeorgou2019, pokal2023improved}; those analyses are limited in
that ozone formation at one location is not confined to that location,
wind transport and atmospheric chemistry spread NO$_x$ over nearby
distances, so emissions at one facility influence ozone at neighbouring
locations \citep{Lu2019, Qu2024}. Untreated locations close to treated ones
may thus experience reductions through interference. This is exactly the
joint SC--SI setting of this paper: neglecting the interference may over- or
under-state the effect of SCR/SNCR and mislead decision-making.

We analyse the $n = 473$ facilities ($152$ treated, $321$ control) of
\citet{DVN/DKXXSN_2016}, previously analysed by \citet{papadogeorgou2019}
and \citet{pokal2023improved}, with fourth-highest daily ozone as the
outcome, SCR/SNCR installation as treatment, $18$ measured covariates, and
facility coordinates; total NO$_x$ emissions, a mediator, is excluded from
the adjustment set. (The distributed data record ozone in ppm; all analyses
here use ppb, the raw outcome multiplied by $10^3$.) We use $\tau = 0.2$ and
caliper $0.25$; sensitivity to both is reported below. Matching rows use
greedy matching at a declared seed; every value in Table~\ref{tab:ozone} is
reproducible from a single documented function call in \textsf{spaci}, and the
scripts in the reproduction repository regenerate the full table.
Uncertainty is quantified as prescribed in
Sections~\ref{sec:idaps} and~\ref{sec:inference}: the matching rows use the
matched-pairs standard error
$\mathrm{sd}\{Y_i - Y_{j(i)}\}/\sqrt{m_T}$, and the doubly robust rows use
the spatial HAC variance computed on the influence values (Bartlett kernel;
data-driven bandwidth, here $b = 2.3^\circ$), with normal critical values
throughout. On these data, the HAC and i.i.d.\ influence-function standard
errors coincide to two decimals (recoverU: $0.347$ vs $0.345$; recoverU$+$:
$0.439$ vs $0.439$), indicating negligible spatial autocorrelation of the
influence values at the observed inter-facility distances; a property of
this application's continental geography, not a general licence for
i.i.d.\ standard errors (Table~\ref{tab:coverage}).

\begin{table}[h]
\centering
\caption{Estimated effect of SCR/SNCR on ambient ozone (ppb), with $95\%$
intervals. Matching rows: greedy matching at the declared seed ($115$), with
matched-pairs standard errors (matched-pair differencing partially cancels
the spatial dependence). Doubly robust rows are deterministic given the data
and use spatial HAC intervals; here the HAC and i.i.d.\ standard errors
coincide to two decimals (HAC bandwidth $2.3^\circ$), indicating negligible
residual spatial correlation of the influence values at the observed
inter-facility distances.}
\label{tab:ozone}
\begin{tabular}{lcc}
\toprule
Method & Estimate & $95\%$ interval \\
\midrule
Naive PS    & \phantom{$-$}1.94 & $(0.22,\ 3.67)$ \\
DAPS        & \phantom{$-$}0.61 & $(-0.81,\ 2.02)$ \\
iDAPS       & $-0.50$ & $(-1.66,\ 0.66)$ \\
recoverU    & $-0.14$ & $(-0.82,\ 0.54)$ \\
recoverU$+$ & $-0.19$ & $(-1.04,\ 0.68)$ \\
\bottomrule
\end{tabular}
\end{table}

\begin{figure}[H]\centering
  \includegraphics[width=0.72\linewidth]{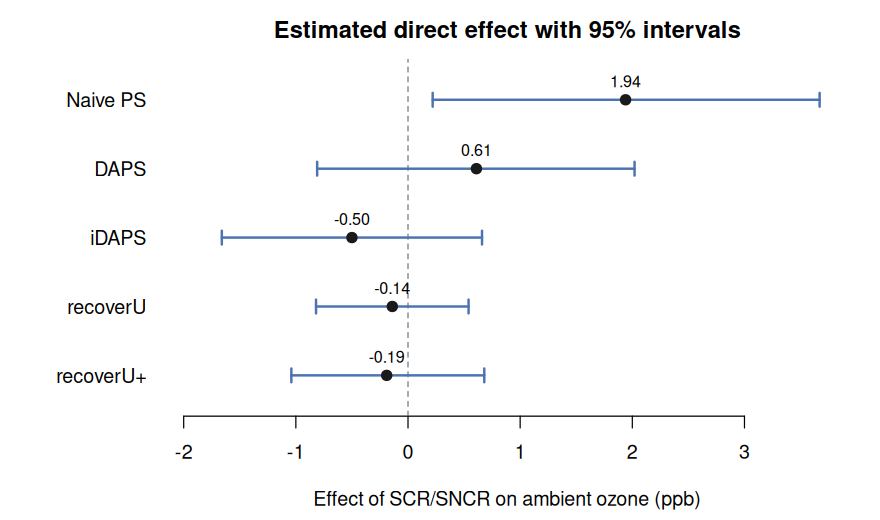}
  \caption{Estimated direct effect of SCR/SNCR installation on ambient ozone
  with $95\%$ intervals (values as in Table~\ref{tab:ozone}: matched-pairs
  standard errors for the matching rows, spatial HAC intervals for the
  doubly robust rows). The naive interval lies above zero; every adjusted
  interval comfortably contains it.}
  \label{fig:ozoneforest}
\end{figure}

Table~\ref{tab:ozone} and Figure~\ref{fig:ozoneforest} show a clear
gradient. The naive analysis suggests SCR/SNCR \emph{increases} ozone by
about $2$ ppb, with an interval excluding zero;  a counterintuitive
conclusion that the adjusted analyses reveal to be an artefact of spatial
confounding: facilities are treated where pollution problems are worst.
Adjusting for the spatial phenomena moves every estimate toward zero, and
all four adjusted intervals comfortably contain it, with iDAPS and the two
doubly robust estimators turning mildly negative. The substantive
conclusion is twofold: there is \emph{no evidence} that SCR/SNCR
installation reduces ambient ozone at the facility level once both spatial
phenomena are accounted for; and a naive analysis would have delivered a
confidently wrong sign. One methodological point in passing: greedy
matching without a declared seed is not reproducible. Across $40$
matching seeds, the point estimates range over $[1.54, 2.25]$ (naive),
$[-0.40, 2.20]$ (DAPS) and $[-0.58, 0.59]$ (iDAPS) which is why
\textsf{spaci} exposes the seed (all matching rows above declare seed $115$)
and offers a deterministic optimal-matching alternative.

Two diagnostics complete the picture. Figure~\ref{fig:matchmap} shows the
geography of the matched samples: naive propensity matching pairs
facilities across the continent, whereas DAPS and iDAPS trade some matches
for far tighter spatial proximity. Exactly the behaviour
Proposition~\ref{prop:bound} rewards, since every kilometre of matched
distance buys confounding bias through the variogram. The plug-in
bias-bound diagnostic for the iDAPS match (coordinatewise form;
$c_{\mathrm{ov}} = 1.2$, calibrated in Section 5.2) evaluates to $X$-term $15.3$, exposure term $0.3$
and confounding term $5.3$ ppb: deliberately conservative,
worst-case-aligned quantities whose message is nonetheless conservative but informative with $127$ matched pairs a mean of $2.1^\circ$
($\approx 230$ km) apart on a residual field with fitted sill $56$
ppb$^2$, continental-scale matching cannot \emph{guarantee} tight
confounder control. This is precisely why we regard the doubly robust
estimates, which use all $473$ units and model the confounder directly, as
the primary analysis, with the matching estimators as design-based
corroboration.
Figure~\ref{fig:psoverlap} shows the estimated propensity-score overlap
between treated and control facilities for the naive, recoverU and
recoverU$+$ propensity models: the overlap is adequate and similar across
the three specifications, so the doubly robust weights are well behaved and
the differences among the estimates in Table~\ref{tab:ozone} reflect
adjustment, not extrapolation.

\begin{figure}[H]\centering
  \includegraphics[width=0.5\linewidth]{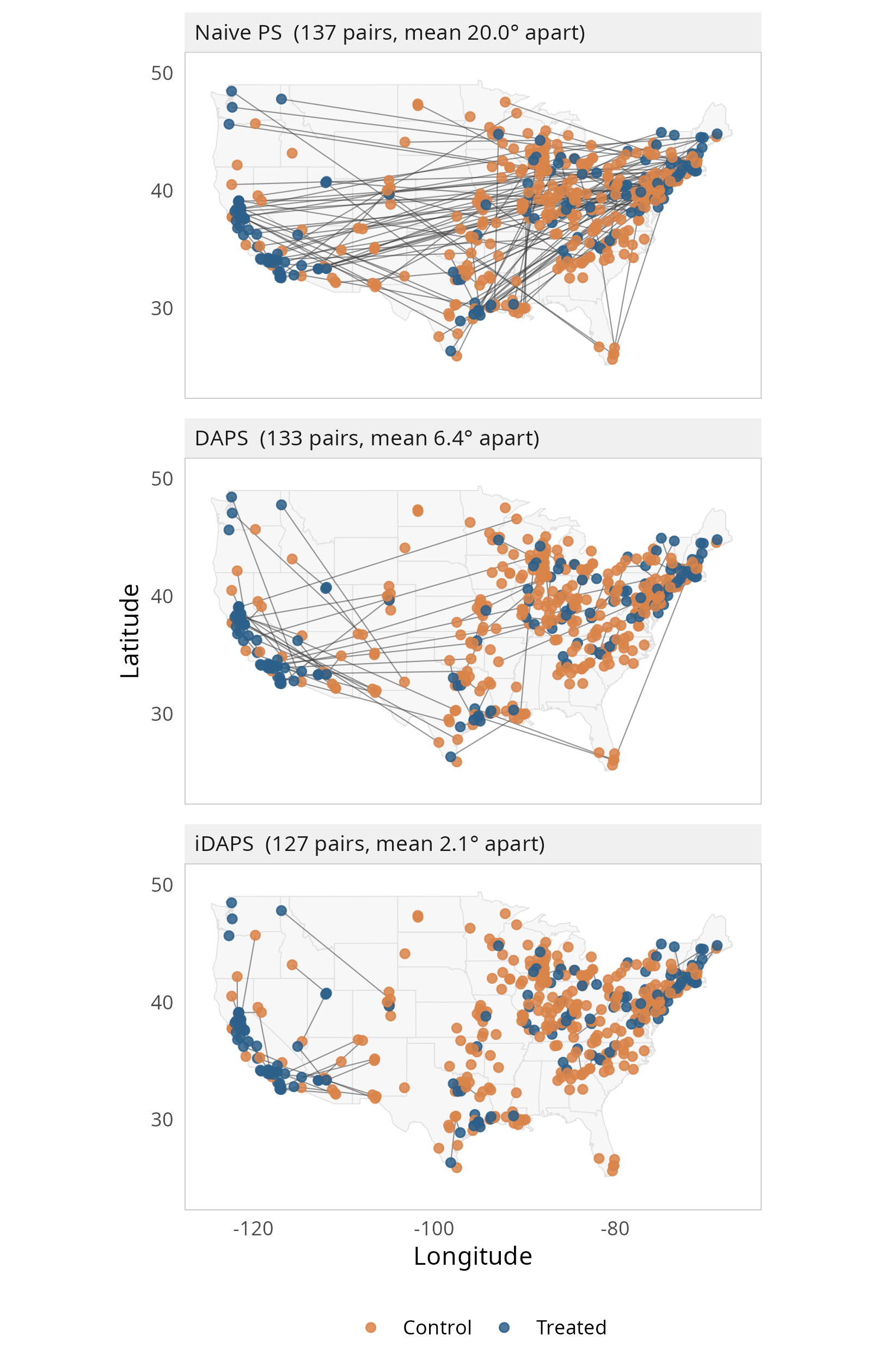}
  \caption{Matched treated (blue) and control (orange) facilities, on a US
  state outline for geographic context, for naive propensity-score matching,
  DAPS and iDAPS (one matching realisation each; panel labels give the
  number of matched pairs and the mean matched distance in degrees). Naive
  matching pairs units across the continent; the distance-adjusted metrics
  keep pairs spatially close, exactly the behaviour rewarded by the
  variogram bias bound of Proposition~\ref{prop:bound}.}
  \label{fig:matchmap}
\end{figure}

\begin{figure}[H]\centering
  \includegraphics[width=0.5\linewidth]{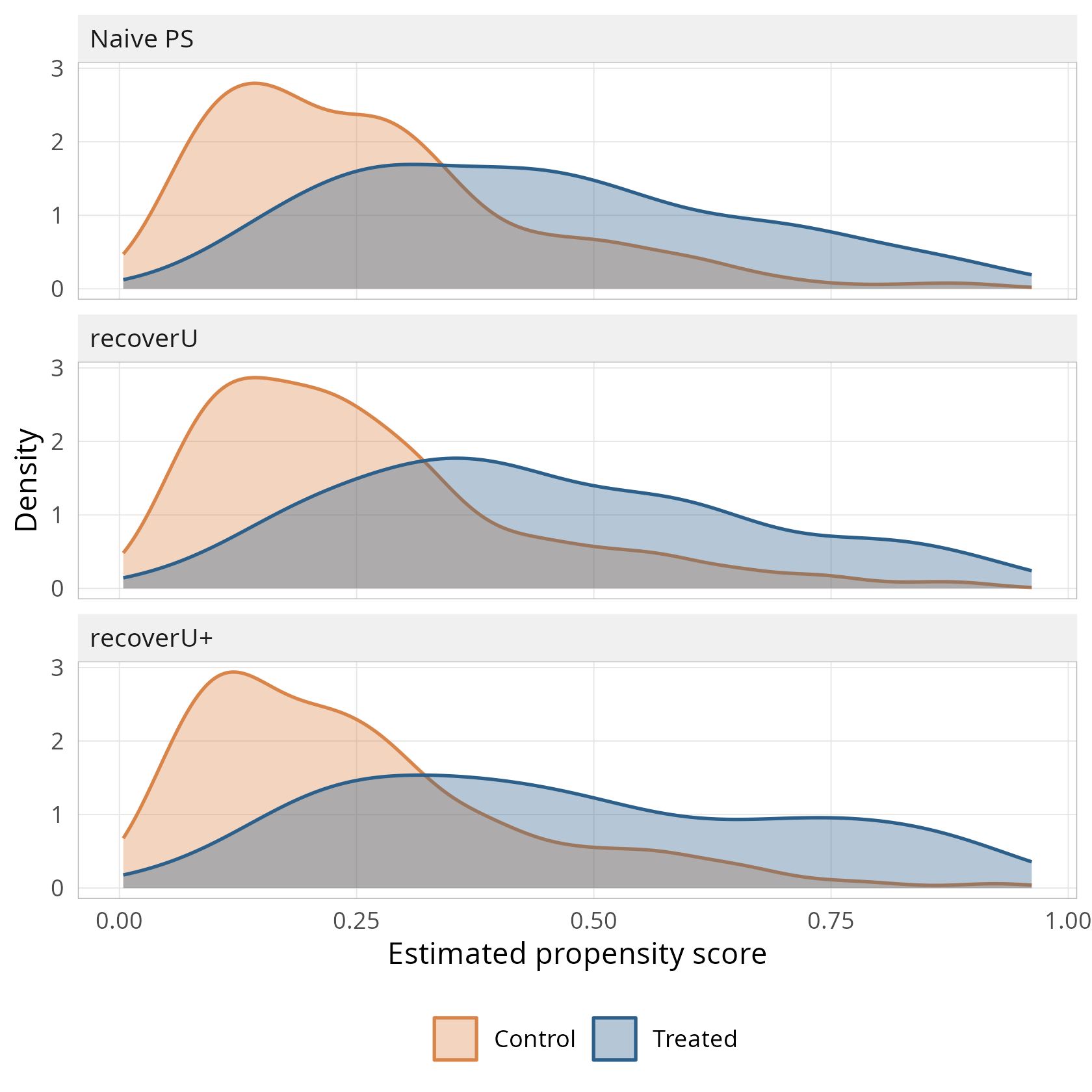}
  \caption{Estimated propensity-score overlap between treated and control
  facilities under the naive, recoverU and recoverU$+$ propensity models.
  Overlap is adequate and stable across specifications, indicating that the
  doubly robust estimates rest on comparison rather than extrapolation.}
  \label{fig:psoverlap}
\end{figure}

Sensitivity: over $\tau \in \{0.1, 0.2, 0.4\}$ the recoverU$+$ estimate
moves within $[-0.19, -0.13]$ and the iDAPS estimate within
$[-0.50, 0.32]$; over calipers $\{0.1, 0.25, 0.5\}$ the iDAPS estimate is
$-0.07$, $-0.50$ and $1.18$. Every configuration's interval contains zero
(marginally at the loosest caliper, where the estimate drifts toward the
naive value exactly as Corollary~\ref{cor:cons} predicts: looser matching
means larger matched distances and less confounding control). The
qualitative conclusion is stable.

\section{Discussion}\label{sec:discussion}

When critical decisions rest on causal effect estimates, even moderate bias
can have consequences ranging from mild to severe, and spatial confounding
and interference are among the most common and most commonly conflated sources of such bias in environmental, agricultural and epidemiological
studies. This paper has given the joint problem a unified, design-based
treatment: a precise estimand ($\mathrm{ADET}$) with an identification result whose
assumptions are stated so that each can be interrogated; a matching
estimator (iDAPS) whose composite metric is not a heuristic but the
minimiser of an explicit bias bound; a doubly robust estimator
(recoverU$+$) whose guarantee is stated honestly, with the unavoidable
residual characterised exactly rather than assumed away; and inference
procedures whose finite-sample calibration is demonstrated rather than
asserted. A distinctive feature of the paper is its validation discipline:
every proposition is accompanied by a numerical experiment designed to
falsify it, and the experiments are reported whether or not they flatter
the methods, the weak-confounding undercoverage of
Section~\ref{sec:sim-calib} being the clearest example.

\subsection*{Practical guidance}

For the applied user, the framework condenses into a workflow.
\emph{(1) Construct the exposure} with a scientifically motivated kernel and
bandwidth, and report a sensitivity band over $\tau$; the stress test of
Section~\ref{sec:misspec} suggests direct-effect estimates are forgiving of
moderate exposure misspecification, but this should be verified per
application. \emph{(2) Run both estimator families.} They fail in different
ways; iDAPS is model-free but pays a price in matched-sample size and in
the variogram bias floor; recoverU$+$ uses all units and models the
confounder directly but inherits $b_{U_{A}}$  so agreement between them is
informative and disagreement is diagnostic. When the units are spread over
distances large relative to the confounder's spatial range (as in our
application, where matched pairs averaged $250$ km apart), the doubly robust
analysis should be primary and the matching analysis corroborative;
in dense designs the preference can reverse.
\emph{(3) Report the diagnostics}: the plug-in bias bound and the mean
matched distance for iDAPS, and the recovery quality for recoverU$+$.
\emph{(4) Never report independence-based standard errors}; use the spatial
HAC (cheap) or the block bootstrap (which also absorbs the matching and
selection steps), and for a headline hypothesis test the conditional
randomisation test, which requires no variance estimation at all.
\emph{(5) Declare the matching seed or use the deterministic optimal
matching}; undeclared greedy matching is not reproducible
(Section~\ref{sec:application}).

\subsection*{What the results teach beyond the two estimators}

Three lessons extend past the specific procedures. First, an
\emph{estimand} lesson: under interference the phrase ``the ATT'' is not
merely imprecise but ill-defined, and analyses that ignore this can compare
estimators that target different quantities without noticing
(Section~\ref{sec:sims}). Defining the direct effect $\mathrm{ADET}$ and the
spillover $\mathrm{SET}$ separately takes only a paragraph, and it removes that
risk. Second,
a \emph{mechanism} lesson: distance-based adjustments work through spatial
differencing, so their currency is the variogram of the confounder at the
matched distance. This single observation explains why DAPS works, why it
degrades under interference (an uncontrolled exposure-imbalance term), why
all such methods share a nugget floor, and it converts matching quality from
a rhetorical claim into a computable number. Third, a \emph{recovered
confounder} lesson: the component of a latent field aligned with treatment
is unrecoverable from a single realisation, not as a practical difficulty
but as a matter of identification. So, ``doubly robust'' claims for
recovered-confounder methods \citep{pokal2023improved} necessarily hold only
up to the residual $b_{U_{A}}$. Making that residual explicit turns an
overclaim into a quantitative sensitivity parameter.

\subsection*{Computation}

The computational profile is asymmetric. iDAPS is light: the distance
matrices are $O(n^2)$ and the weight search vectorises, so even the full
grid completes in seconds at $n = 500$. recoverU$+$ is dominated by the
$O(n^3)$ Mat\'ern likelihood evaluations of the recovery stage
($\approx 15$ s per fit at $n = 300$); fixing the smoothness at $\nu = 1/2$
both stabilises the fit and reduces its cost, and for $n$ in the thousands
the dense solves should be replaced by Vecchia-type approximations, which
slot into the recovery step without altering the estimator. One
computational finding deserves wide circulation: when a resampling scheme
contains a kriging stage, resampling \emph{observations} duplicates
locations and the kriging over-smooths them, deflating the bootstrap
variance below even the i.i.d.\ one; resampling the influence values avoids
the pathology (Section~\ref{sec:inference}).

\subsection*{Limitations}

First, recoverU$+$ carries the residual bias $b_{U_{A}}$
(Proposition~\ref{prop:pdr}); it vanishes only when the confounder does not
drive treatment, and its closed form under a probit treatment model is the
natural anchor for a sensitivity analysis, a ``spatial E-value'' in the
sense of \citet{vanderweele2017} that we leave to future work.
Second, iDAPS inherits an irreducible bias floor from micro-scale (nugget)
variation in the confounder (Corollary~\ref{cor:cons}); no distance-based
method escapes it, and reporting the fitted nugget alongside the bound makes
the floor visible.

Third, the \emph{recovery-variance gap} when spatial confounding is
weak, \emph{all} influence-based standard errors undercover, and the reason
is structural rather than small-sample: in the $u = 0$ rows of
Table~\ref{tab:coverage} the SE ratios are $\approx 0.55$ and coverage
$\approx 0.75$, essentially unchanged from $n = 150$ to $n = 300$. The
mechanism is a \emph{generated-regressor} problem
\citep{pagan1984, murphy1985}. recoverU$+$ is a two-stage estimator: it
first \emph{constructs} the covariate $\widehat U_{R}$ from the data by
kriging, and then treats it as if it were an ordinary observed covariate.
Every standard error of Section~\ref{sec:inference} is computed conditional
on the one realised $\widehat U_{R}$, so the sample-to-sample variability of
the recovery step itself is invisible to all of them. When confounding is
strong this conditioning is harmless, $\widehat U_{R}$ is anchored by
genuine spatial signal and barely varies across samples. When confounding
is weak there is little signal to recover, $\widehat U_{R} $ is largely
re-fitted noise that changes from sample to sample, and the variability it
passes into $\hat\tau$ is exactly the component the standard errors miss.
Neither implemented remedy reaches it (the influence-resampling bootstrap
conditions on $\widehat U_{R}$ by construction; re-running the full pipeline
on resampled blocks triggers the duplicate-coordinates pathology noted
under \emph{Computation}), but a \emph{parametric field bootstrap}, 
simulating residual fields from the fitted Mat\'ern at the original
locations and repeating the recovery and estimation on each replicate
would, and we regard it as the most practical correction.
Recovery-variance-corrected inference is, to our knowledge, an open
problem; the practical cost of the gap is contained, because the regime in
which it bites, weak spatial confounding, is precisely the regime in which
recoverU$+$ is least needed.

\subsection*{Extensions}

Several extensions of the framework are natural. First, the framework could be extended to mediation in the presence of both spatial confounding and spatial interference. The ozone application already provides a natural setting for this extension because total NOx emissions can be viewed as a mediator between SCR/SNCR technology and ambient ozone. This would allow the total effect to be decomposed into direct and NOx-mediated pathways while accounting for the same spatial confounding and interference that motivate the present analysis.

Second, the framework could accommodate continuous treatments through generalised propensity scores \citep{giffin2023generalized}, allowing the methodology to address settings in which treatment intensity, rather than treatment status, varies across locations. A further extension would consider spatio-temporal interventions, such as staggered technology adoption. In such settings, spatial differencing could address spatially structured confounding while temporal differencing, as in difference-in-differences designs, could remove additional time-varying confounding components.

Finally, the theoretical and inferential properties of recoverU+ warrant further development. In particular, a formal central limit theorem under increasing-domain asymptotics would provide a rigorous basis for the influence-function inference used here. In the weak-confounding regime, the current influence-function standard errors do not account for the sample-to-sample variability introduced by estimating the recovered confounder. A parametric field bootstrap that 
simulates the spatial confounder from the fitted Matérn model and re-runs the recovery and estimation stages could provide a practical way to propagate this additional uncertainty. Future work could also compare the proposed frequentist framework with joint Bayesian approaches to spatial confounding and interference, particularly in terms of sensitivity to modelling assumptions and prior specification.

\section*{Code availability}
Two openly available components accompany this paper: (i) the R package
\textsf{spaci}, which implements all of the methodology
(\url{https://cran.r-project.org/web/packages/spaci/}); and (ii) a reproduction
repository containing the scripts that regenerate every table and figure in
this paper, with a script-to-output map and runtime guidance
(\url{https://github.com/olatunjijohnson/spaci-paper-reproduction}).

\section*{Data availability}
The power-plant facility data are publicly available from the Harvard
Dataverse \citep{DVN/DKXXSN_2016} and are bundled, with documentation, in
the \textsf{spaci} package. The distributed ozone outcome is recorded in ppm;
the analyses in this paper use ppb (the raw values multiplied by $10^3$).

\appendix

\section*{Appendix}
\section{Proof of Proposition~\ref{prop:id} (identification)}
\label{app:id}

The proof proceeds in four steps, each converting one non-observable object
into an observable one; a different assumption licenses each step.

\emph{Step 1: from interference to the exposure summary.}
Fix a treatment level $a$ and an exposure level $e$. Under interference the
potential outcome $Y_i(a, e)$ could in principle depend on the entire
treatment vector $\bm{A}_{-i}$ of other units; Assumption~\ref{a:exp} states
that it does not, all dependence on others' treatments flows through the
scalar summary $E_i$. Consequently, conditioning on $E_i = e$ makes
$\bm{A}_{-i}$ irrelevant, and together with consistency
(Assumption~\ref{a:con}, $Y_i = Y_i(A_i, E_i)$) the conditional mean of the
\emph{observed} outcome among units with $A_i = a$ and $E_i = e$ equals the
conditional mean of the potential outcome:
\[
  \mathbb{E}[Y \mid A = a,\ X,\ E = e,\ U]
  \;=\; \mathbb{E}[Y(a, e) \mid A = a,\ X,\ E = e,\ U].
\]
This is the step that would fail under a misspecified exposure mapping,
which is why Assumption~\ref{a:exp} is stated separately and stress-tested
in Section~\ref{sec:misspec}.

\emph{Step 2: removing the conditioning on treatment.}
Assumption~\ref{a:ign} states that, given $(X, E, U)$, the treatment carries
no further information about the potential outcomes:
$\{Y(a,e)\} \perp \!\!\! \perp A \mid X, E, U$. Hence the right-hand side of Step 1
equals $\mathbb{E}[Y(a,e) \mid X, E = e, U]$, the same quantity for treated and
control units. Assumption~\ref{a:pos} guarantees that both treatment arms
occur with positive probability at every $(X, E, U)$, so the $a = 0$
conditional mean is well defined on the support of the treated units, which
is where $\mathrm{ADET}$ averages.

\emph{Step 3: replacing the latent confounder by its recoverable part.}
Steps 1--2 identify the potential-outcome means given $(X, E, U)$, but $U$
is unobserved. Assumption~\ref{a:rec} states that the control-arm mean does
not change when the conditioning variable $U$ is replaced by $U_{R}$:
$\mathbb{E}[Y(0,e) \mid X, E, U] = \mathbb{E}[Y(0,e) \mid X, E, U_{R}]$. (For the treated arm
no such assumption is needed: on the event $A = 1$ the observed outcome
\emph{is} $Y(1, E)$, so $\mathbb{E}[Y(1,E) \mid A{=}1, V]$ is directly observable.)
Writing $V = (X, E, U_{R})$, the two conditional means in the statement of the
proposition are therefore observable functionals.

\emph{Step 4: assembling the estimand.}
By iterated expectations over the distribution of $V$ among the treated,
\begin{align*}
     \mathrm{ADET}
  &= \mathbb{E}\big\{ \mathbb{E}[Y(1,E) \mid A{=}1, V] - \mathbb{E}[Y(0,E) \mid A{=}1, V]
      \,\big|\, A = 1 \big\} \\
  &= \mathbb{E}\big\{ \mathbb{E}[Y \mid A{=}1, V] - \mathbb{E}[Y \mid A{=}0, V] \,\big|\, A=1 \big\},
\end{align*}

where the first term uses consistency directly and the second chains Steps
1--3. This proves part (a).

\emph{Part (b).}
Suppose Assumption~\ref{a:rec} fails but the additive working model
\eqref{eq:outcome} holds, so that
$\mathbb{E}[Y(0,e) \mid X, E, U] = \mu_0(V) + \theta_U U_{A}$ with $\mu_0$
$V$-measurable. Repeating Step 3 with this expression, the identified
functional becomes
\[
  \mathbb{E}\big\{ \mathbb{E}[Y \mid A{=}1, V] - \mathbb{E}[Y \mid A{=}0, V] \,\big|\, A{=}1 \big\}
  = \mathrm{ADET} + \theta_U\,
    \mathbb{E}\big\{ \mathbb{E}[U_{A} \mid A{=}1, V] - \mathbb{E}[U_{A} \mid A{=}0, V]
      \,\big|\, A = 1 \big\}.
\]
It remains to show the correction term equals $b_{U_{A}}$ of
Proposition~\ref{prop:pdr}. The first summand is
$\mathbb{E}[U_{A} \mid A = 1]$ by iterated expectations. For the second, write
$m(V) = \mathbb{E}[U_{A} \mid A = 0, V]$ and let $p(V) = \mathbb{P}(A{=}1 \mid V)$ with
odds $w = p/(1-p)$; then
\begin{align*}
  \mathbb{E}[m(V) \mid A{=}1]
  &= \frac{\mathbb{E}[A\, m(V)]}{\mathbb{E}[A]}
  = \frac{\mathbb{E}[p(V)\, m(V)]}{\mathbb{E}[A]}
  = \frac{\mathbb{E}[(1-A)\, w(V)\, m(V)]}{\mathbb{E}[A]} \\
  &= \frac{\mathbb{E}[(1-A)\, w(V)\, U_{A}]}{\mathbb{E}[(1-A) w(V)]}
  = \mathbb{E}^{w}[U_{A} \mid A{=}0],
\end{align*}
using $\mathbb{E}[A \mid V] = p(V)$, $\mathbb{E}[(1-A) \mid V] = 1 - p(V)$, the definition
of $m$, and $\mathbb{E}[(1-A) w] = \mathbb{E}[p(V)] = \mathbb{E}[A]$. Hence, the correction is
exactly $\theta_U(\mathbb{E}[U_{A} \mid A{=}1] - \mathbb{E}^{w}[U_{A} \mid A{=}0]) = b_{U_{A}}$.
\qed

\section{Proof of Proposition~\ref{prop:bound} and
Corollary~\ref{cor:cons}}
\label{app:bound}

\emph{The exact decomposition \eqref{eq:decomp-bias}.}
Let $(i, j(i))$ denote a matched pair with $A_i = 1$, $A_{j(i)} = 0$.
Differencing the working model \eqref{eq:outcome} within the pair,
\[
  Y_i - Y_{j(i)}
  = \theta_1
  + \theta_2^{\top}(X_i - X_{j(i)})
  + \theta_3 (E_i - E_{j(i)})
  + \theta_U \{U(s_i) - U(s_{j(i)})\}
  + (\varepsilon_i - \varepsilon_{j(i)}),
\]
because the treatment terms contribute $\theta_1 \cdot 1 - \theta_1 \cdot 0
= \theta_1$. Averaging over the $m_T$ pairs and taking expectations yields
\eqref{eq:decomp-bias} provided
$\mathbb{E}[\varepsilon_i - \varepsilon_{j(i)}] = 0$. This last point deserves care,
because the pair indices are random: the matching map is a deterministic
function of the propensity scores, locations and exposures, hence
measurable with respect to $\sigma(\bm{A}, \bm{X}, \bm{s})$; since
$\varepsilon \perp \!\!\! \perp (\bm{A}, \bm{X}, \bm{s}, U)$ with mean zero, the
selected differences $\varepsilon_i - \varepsilon_{j(i)}$ remain mean-zero
conditional on the matching, and the unconditional expectation vanishes.

\emph{The bound \eqref{eq:bound}: covariate and exposure terms.}
By the triangle inequality and Jensen's inequality
($|\mathbb{E}[Z]| \le \mathbb{E}|Z|$ applied coordinatewise),
\[
  \big|\theta_2^{\top} \mathbb{E}[\Delta X]\big|
  \le \sum_{k=1}^{p} |\theta_{2,k}|\, \big|\mathbb{E}[\Delta X_k]\big|
  \le \sum_{k=1}^{p} |\theta_{2,k}|\, \mathbb{E}|\Delta X_k| ,
  \qquad
  |\theta_3\, \mathbb{E}[\Delta E]| \le |\theta_3|\, \mathbb{E}|\Delta E| .
\]
The coordinatewise form is preferred to the Cauchy--Schwarz aggregate
$\lVert\theta_2\rVert \mathbb{E}\lVert\Delta X\rVert$ because each product
$|\theta_{2,k}|\,\mathbb{E}|\Delta X_k|$ is expressed in outcome units and is
therefore invariant to rescaling individual covariates; the aggregate is
not, and can be vacuous when covariates are heterogeneously scaled (as in
the application, where covariate scales span several orders of magnitude).

\emph{The bound \eqref{eq:bound}: confounding term.}
First suppose the matching map were independent of $U$ given the locations.
Conditional on a pair at distance $d$, stationarity gives
$\mathbb{E}[(U_i - U_j)^2 \mid d] = 2\gamma_U(d)$ by the definition of the
semivariogram, and Jensen's inequality
($\mathbb{E}|Z| \le \sqrt{\mathbb{E}[Z^2]}$) yields
$\mathbb{E}[|\Delta U| \mid d] \le \sqrt{2\gamma_U(d)}$; iterating expectations over
the matched-distance distribution gives
$\mathbb{E}|\Delta U| \le \mathbb{E}[\sqrt{2\gamma_U(d_{ij(i)})}]$.
In fact the matching map is \emph{not} independent of $U$: it is
$\sigma(\bm{A}, \bm{X}, \bm{s})$-measurable, and $\bm{A}$ depends on $U$
through the treatment model, so matched pairs are selected in a way that
tilts the conditional law of $U$. The tilt is controlled by positivity:
by Bayes' rule the density ratio of $U$ given $\bm{A}$ to its marginal is a
ratio of treatment-assignment probabilities, which Assumption~\ref{a:pos}
bounds away from $0$ and $1$; hence there exists a finite constant
$c_{\mathrm{ov}} \ge 1$ with
$\mathbb{E}[|\Delta U| \,\big|\, \text{matched}] \le
c_{\mathrm{ov}}\, \mathbb{E}[\sqrt{2\gamma_U(d_{ij(i)})}]$. The constant is
finite-sample in nature; empirically $c_{\mathrm{ov}} \approx 1.2$ across
the validation grid of Section~\ref{sec:sim-bound}, and the uniform-in-$n$
theory of this constant is part of the companion asymptotic work.
Multiplying by $|\theta_U|$ and combining the three displays proves
\eqref{eq:bound}. \qed

\emph{Proof of Corollary~\ref{cor:cons}.}
If the caliper $c_n \to 0$ while $m_T \to \infty$, the matched distances
satisfy $d_{ij(i)} \to 0$ (the composite metric dominates a multiple of the
normalised spatial distance whenever $\pi_2 > 0$, and the balance criterion
penalises $\bar d$ directly). If $U$ is mean-square continuous with
$\gamma_U(0^{+}) = 0$, then $\sqrt{2\gamma_U(d_{ij(i)})} \to 0$ pointwise
and, being bounded by $\sqrt{2\,\sup_h \gamma_U(h)} < \infty$, the
expectation converges to zero by dominated convergence; the confounding
term of \eqref{eq:bound} vanishes. If instead $U$ has a nugget, i.e.\
$\gamma_U(h) \ge \eta > 0$ for all $h > 0$, then
$\mathbb{E}[\sqrt{2\gamma_U(d)}] \ge \sqrt{2\eta}$ \emph{regardless} of how small
the matched distances become, giving the floor
$c_{\mathrm{ov}} |\theta_U| \sqrt{2\eta}$. The floor is shared by any
method whose confounding control operates through spatial proximity,
because micro-scale variation is precisely the component of $U$ that
proximity cannot difference away. \qed

\section{Proof of Proposition~\ref{prop:pdr} (partial double robustness)}
\label{app:pdr}

Throughout, expectations are population quantities; $p(V)$ denotes the
working propensity model, $w = p/(1-p)$ its odds, $m_0(V)$ the working
control-outcome model, and $\mu_0(V) = \mathbb{E}[Y(0) \mid V]$ the true one. The
population version of the estimator \eqref{eq:dr} is
\[
  \tau^{*}
  = \frac{1}{\pi_1}\,
    \mathbb{E}\Big[\big\{A - (1-A)\, w(V)\big\}\big(Y - m_0(V)\big)\Big],
  \qquad \pi_1 = \mathbb{E}[A].
\]

\emph{Step 1: substitute the outcome decomposition.}
Under \eqref{eq:outcome} with Assumption~\ref{a:rec} relaxed to the additive
form, $Y = \theta_1 A + \mu_0(V) + \theta_U U_{A} + \tilde\varepsilon$ with
$\mathbb{E}[\tilde\varepsilon \mid V, A] = 0$. Then
\[
  Y - m_0(V) = \theta_1 A + \{\mu_0 - m_0\}(V) + \theta_U U_{A}
             + \tilde\varepsilon .
\]

\emph{Step 2: the treatment term identifies the estimand.}
The contribution of $\theta_1 A$ is
$\pi_1^{-1} \mathbb{E}[\{A - (1-A) w\}\, \theta_1 A]
= \pi_1^{-1} \theta_1 \mathbb{E}[A] = \theta_1 = \mathrm{ADET}$,
because $A(1-A) = 0$ annihilates the weighted-control part.

\emph{Step 3: the noise term vanishes.}
$\mathbb{E}[\{A - (1-A)w(V)\}\,\tilde\varepsilon]
= \mathbb{E}\big[\{A - (1-A)w(V)\}\, \mathbb{E}[\tilde\varepsilon \mid V, A]\big] = 0$,
since the bracket is $(V, A)$-measurable.

\emph{Step 4: the model-error term vanishes if either nuisance is correct.}
Write $g(V) = \{\mu_0 - m_0\}(V)$, a $V$-measurable error. If the outcome
model is correct, $g \equiv 0$ and the term vanishes trivially. If instead
the propensity model is correct, $w$ uses the true $p(V)$, and the key
identity is
\begin{align*}
  \mathbb{E}[(1-A)\, w(V)\, g(V)]
  &= \mathbb{E}\big[\mathbb{E}[1-A \mid V]\, w(V)\, g(V)\big] \\
 & = \mathbb{E}\big[(1 - p(V))\, \tfrac{p(V)}{1-p(V)}\, g(V)\big] \\
 & = \mathbb{E}[p(V)\, g(V)] \\
  &= \mathbb{E}[A\, g(V)] ,
\end{align*}
where the first equality is iterated expectations and the last uses
$\mathbb{E}[A \mid V] = p(V)$ again. Hence
$\mathbb{E}[\{A - (1-A)w\}\, g(V)] = \mathbb{E}[A g] - \mathbb{E}[A g] = 0$: the odds-weighted
controls exactly reproduce the treated distribution of any
$V$-measurable function --- this is the double-robustness mechanism.

\emph{Step 5: the surviving term.}
No analogous cancellation applies to $\theta_U U_{A}$, because $U_{A}$ is not
$V$-measurable, and by construction of the projection \eqref{eq:decomp} it
is correlated with $A$ even given $V$. Collecting Steps 1--4,
\[
  \tau^{*} - \mathrm{ADET}
  = \frac{\theta_U}{\pi_1}
    \Big( \mathbb{E}[A\, U_{A}] - \mathbb{E}[(1-A)\, w(V)\, U_{A}] \Big)
  = \theta_U \Big( \mathbb{E}[U_{A} \mid A{=}1]
    - \mathbb{E}^{w}[U_{A} \mid A{=}0] \Big)
  = b_{U_{A}},
\]
using $\mathbb{E}[A U_{A}] = \pi_1 \mathbb{E}[U_{A} \mid A{=}1]$ and, at the true $p$,
$\mathbb{E}[(1-A)w] = \mathbb{E}[p(V)] = \pi_1$.

\emph{Claim (i).}
If $U \perp \!\!\! \perp \bm{A}$, the $L_2$-projection of $U$ onto the span of the
(centred) treatment field is zero, so $U_{A} \equiv 0$ and $b_{U_{A}} = 0$:
the classical double-robustness statement is recovered as the special case
of no unmeasured treatment--confounder dependence.

\emph{Claim (ii).}
By Cauchy--Schwarz applied to each term,
$|\mathbb{E}[A U_{A}]| \le \sqrt{\pi_1}\, \sigma_{U_{A}}$ and
$|\mathbb{E}[(1-A) w U_{A}]| \le \big(\mathbb{E}[(1-A) w^2]\big)^{1/2} \sigma_{U_{A}}$, so
\[
  |b_{U_{A}}| \;\le\; |\theta_U|\, \sigma_{U_{A}}\, \kappa,
  \qquad
  \kappa = \pi_1^{-1/2} + \pi_1^{-1}\big(\mathbb{E}[(1-A)\, w(V)^2]\big)^{1/2},
\]
a constant depending only on the treatment prevalence and the tail of the
odds weights i.e.\ on overlap. Poor overlap inflates
$\mathbb{E}[(1-A)w^2]$ and hence the worst-case residual bias, which is the same
mechanism that drives the finite-sample fragility observed in the CM cell
of Table~\ref{tab:dr}. \qed

\bibliographystyle{plainnat}

\bibliography{sn-bibliography}

\end{document}